\documentclass[10pt]{amsart}
\makeatletter
\renewcommand{\@secnumfont}{\bfseries}

\renewcommand{\section}{\@startsection{section}{1}%
  \z@{.7\linespacing\@plus\linespacing}{-.5em}%
  {\normalfont\bfseries}}

\def\@seccntformat#1{\indent\S\csname the#1\endcsname.\ }
\makeatother

\usepackage{amsmath,amsthm,amsfonts,amssymb,graphicx,hyperref}
\usepackage[theoremfont,helvratio=1]{newtxtext}
\usepackage[
  paperwidth=171.325mm,   %
  paperheight=253.98mm,  %
  textwidth=30pc,         %
  textheight=579.2pt,     %
  left=21mm,              %
  top=23.3mm,             %
  headheight=6pt,
  headsep=18pt,
  footskip=26pt,
  columnsep=10pt
]{geometry}

\usepackage[english]{babel}
\usepackage{verbatim}
\usepackage{enumitem}
\usepackage{framed}
\usepackage{stmaryrd,pifont}
\usepackage[numbers]{natbib}
\usepackage{mathtools}
\usepackage{cancel} %
\usepackage[utf8]{inputenc} %
\usepackage{xcolor}
\usepackage[all]{xy} \SilentMatrices
\usepackage{forest}
\usepackage{tikz}
\usetikzlibrary{matrix,arrows,bending, positioning, calc}

\usepackage{graphicx} 
\newcommand{\minus}{\smallsetminus}

\newcommand{\defi}{\mathrel{\mathop:}=} %
\newcommand{\calA}{\mathcal{A}}
\newcommand{\R}{\mathbb{R}}
\newcommand{\N}{\mathbb{N}}

\newcommand{\Q}{\mathbb{Q}}

\newcommand{\Borel}{\mathcal{B}}
\newcommand{\Power}{\mathcal{P}}
\renewcommand{\emptyset}{\varnothing}

\newcommand{\Rel}{\mathcal{R}}

\newcommand{\sem}[1]{{\llbracket #1 \rrbracket}} %
\newcommand{\pos}[1]{{\langle #1 \rangle}} %

\newcommand{\lmp}[1]{\mathbb{#1}}

\newcommand{\union}{\mathop{\textstyle\bigcup}}

\newcommand{\sm}{\minus}
\newcommand{\sbq}{\subseteq}

\usepackage{todonotes}

\newcommand{\leb}{\mathfrak{m}}

\DeclareMathSymbol{\rest}{\mathord}{AMSa}{"16}

\newcommand{\Tfont}[1]{\mathsf{#1}}
\newcommand{\rR}{\mathrel{R}} %

\newcommand{\s}{\mathsf{s}}
\newcommand{\sfs}{\mathsf{s}}
\newcommand{\e}{\mathsf{e}}
\newcommand{\h}{\mathsf{h}}
\newcommand{\Fam}{\mathcal{F}}
\newcommand{\D}{\mathcal{D}}

\newcommand{\liftrel}[1]{\widehat{#1}}

\newcommand{\dom}{\mathrm{dom}}
\newcommand{\id}{\mathrm{id}}

\DeclareMathOperator{\inl}{inl}
\DeclareMathOperator{\inr}{inr}
\DeclareMathOperator{\supp}{supp}

\newcommand{\Suc}{\mathit{Suc}}
\newcommand{\Tr}{\mathit{Tr}}
\newcommand{\LTS}{\mathit{LTS}}
\newcommand{\WF}{\mathit{WF}}
\newcommand{\bigd}{\Diamond}

\newcommand{\ltsfrom}{\mathbb{T}}

\newlength{\xywd}
\newcommand{\xyrightarrow}[2][]{%
	\sbox{0}{$\scriptstyle#1$}%
	\xywd=\wd0
	\sbox{0}{$\scriptstyle#2$}%
	\ifdim\wd0>\xywd \xywd=\wd0 \fi
	\xymatrix@C\dimexpr\xywd+1em\relax{{}\ar[r]^{#2}_{#1}&{}}%
}
\newcommand{\Rlab}[1]{\xyrightarrow{#1}}

\makeatletter
\newcommand{\oset}[3][0ex]{%
	\mathrel{\mathop{#3}\limits^{
			\vbox to#1{\kern-2\ex@
				\hbox{$\scriptstyle#2$}\vss}}}}
\makeatother
\newcommand{\RlabLTS}[1]{\oset[.2ex]{#1}{\longrightarrow}}

\makeatletter
\let\lemma\@undefined
\let\endlemma\@undefined 
\let\definition\@undefined
\let\enddefinition\@undefined 
\let\example\@undefined
\let\endexample\@undefined 
\let\remark\@undefined
\let\endremark\@undefined
\let\corollary\@undefined
\let\endcorollary\@undefined
\makeatother

\newtheorem{theorem}{Theorem}[section]
\newtheorem{lemma}[theorem]{Lemma}
\newtheorem{prop}[theorem]{Proposition}
\newtheorem{corollary}[theorem]{Corollary}

\newtheorem*{claim*}{Claim}
\theoremstyle{definition}
\newtheorem{definition}[theorem]{Definition}
\newtheorem*{definition*}{Definition}
\newtheorem{remark}[theorem]{Remark}

\newtheorem{example}[theorem]{Example}
\theoremstyle{remark}

\newtheorem*{example*}{Example}

\newcommand{\xindex}[2][{}]{\relax}
\usepackage{orcidlink}

\hypersetup{
  colorlinks,
  urlcolor={blue},
  linkcolor={blue!50!black},
  citecolor={blue!50!black},
}

\begin{document}

\title{The complexity of bisimilarity on pointmass processes%
}
\author{Martín Santiago Moroni%
		\orcidlink{0009-0007-3542-5665}}		
\address{Facultad de Ciencias Exactas y Naturales.
	Departamento de Computación, Universidad de Buenos Aires and ICC UBA \& CONICET.}
\email{martin.moroni@unc.edu.ar}
\author{Pedro Sánchez Terraf%
    	\orcidlink{0000-0003-3928-6942}}
\address{Universidad Nacional de Córdoba.  Facultad de Matemática, Astronomía,  Física y Computación.
   Centro de Investigación y Estudios de Matemática (CIEM-FaMAF), Conicet. Córdoba. Argentina.}
\email{psterraf@unc.edu.ar}   
\thanks{The second author's work was supported by Secyt-UNC project 33620230100751CB and Conicet PIP project 11220210100508CO}


\begin{abstract}
  We assess the descriptive complexity of \textit{bisimilarity} or “equality of
  behavior” on a family of Markov decision processes over uncountable standard Borel
  spaces, namely \textit{nondeterministic labelled Markov processes} (NLMP).

  We show that bisimilarity is analytic for processes with a uniform assignment of
  finitely-supported measures to each state and label. More precisely, we obtain
  that bisimilarity on the space of countable Kripke frames (or labelled
  transition systems) is classifiable by countable structures.

  We show that bisimilarity of well-founded (“terminating”) processes is
  Borel. We also provide a lower complexity bound by reducing the relation of
  eventual equality of binary sequences, $E_0$, to the former. As a consequence,
  there is no countable fragment of basic modal logic with denumerable
  conjunctions that characterizes bisimilarity for processes of small
  rank.

  We finally apply the previous Borel-definability result to study the well-founded
  part of discrete uniform processes over continuous spaces.
\end{abstract}

\maketitle 

\section{Introduction}
\label{sec:introduction}
This paper is part of an ongoing project to assess the descriptive complexity of “equality of behavior” on computational processes; more precisely, of Markov
decision processes over uncountable spaces.

The main motivation for studying such spaces (for which the informal term
“continuous”, referring to the topological or measurable structure, is customary) stems from an expanding
vision of the concept of \textit{process} in which we include, for example, the
physical components or mobile parts that are controlled by a computer. As a
consequence, the “state space” of such a process combines all variables involved
in this expanded framework. Ideally, we wish to describe the process in
its original presentation, prior to any discretization procedure that might be
needed for explicit calculations.

Processes as described in the previous paragraph often include several modes of
indeterminateness, which come in at least two flavors: probabilistic and
nondeterministic, and they might appear in combination.

The probabilistic part models a “quantified uncertainty” of
sorts: For instance the reliability with which the aforementioned components comply with the
controller's orders; and also, that arising from randomized computations, like a
Monte Carlo approximation. Examples in the literature involving this
probabilistic component as a main ingredient are \emph{labelled Markov processes (LMP)} \cite{Desharnais,DEP} and stochastic Kripke frames which have
\emph{stochastic \emph{accessibility} relations}
\cite{doberkat2007stochastic}.

The second, more opaque type of uncertainty is not directly related to probability: “internal nondeterminism,” about which little to no information is available.
This form of indetermination might represent our ignorance of the respective probabilities, or just our
desire to describe an abstract general situation that applies to several
concrete models (“subespecification”). A third example, which originated the concept of nondeterminism in
\textit{concurrent computation}, is the variety of behaviors arising when a
“scheduler” sequentially organizes a series of tasks that were issued in
parallel by different “cores” of the computer.

The most general setting in this work will be that of \textit{nondeterministic
labelled Markov processes} (NLMP)
\cite{Wolovick,D'Argenio:2012:BNL:2139682.2139685} whose internal nondeterminism
is represented by a menu of distinct probabilistic behaviors for each
state and chosen action. In this work we will focus on processes with countable
menus in which all (sub)probability measures appearing are finitely supported.
NLMPs carry a measurable structure that ensures that
basic related concepts (such as the validity sets of the appropriate modal logic)
are definable in the descriptive set-theoretic sense.

Prior structures from the literature can be considered as particular cases of NLMP. Kripke
frames, for instance, faithfully represent processes with no probabilistic part;
when interpreted as computational processes, they are called \textit{labelled
  transition systems} (LTS).
Kripke frames/LTSs also carry a natural notion of computational equivalence,
\textit{bisimilarity}.
This fundamental notion plays an important role in many areas of Computer Science.
This is reflected by its diverse characterizations from various perspectives,
such as fixed point and coinduction theory \cite{sangiorgi}, coalgebras, modal
logics \cite{doberkat2009stochastic}, and game theory.
Bisimilarity was first developed for discrete processes \cite{sangiorgi09}, and
later adapted to include probabilistic and nondeterministic systems. It is not straightforward to generalize this notion all
the way through NLMP, and there are many natural contenders, which have been
discussed elsewhere \cite{coco,survey_bisim}. Luckily, many diverging versions
coincide in the restricted context we are considering here.

A key aspect of bisimilarity is its  connection to
modal logic, which originates in the work of van Benthem \cite{vanBenthem1976}
and Hennessy and Milner \cite{HennessyMilner85:jacm}. An appropriate logical
characterization always allows one to conclude a certain degree of descriptive “tameness” of
bisimilarity. In the discrete probabilistic setting the former was achieved by Larsen
and Skou
\cite{LarsenSkou} and for continuous state spaces by Desharnais et
al. \cite{Des98LogicalChar,Desharnais}. Further references are Clerc et
al.~\cite{ic19}, and the second author's \cite{Pedro20111048,
  2012arXiv1211.0967S} which contain preliminary
results on the study of the descriptive complexity of bisimilarity.

In the present work, we continue the line of the previously mentioned papers. We have a progression of results, involving increasingly stronger
hypotheses. The most general one involves Markov decision processes over
standard Borel spaces:
\begin{theorem}[Theorem~\ref{th:state-bisim-finit-supp}]\label{th:bisim-unif-finit-supp}
  Bisimilarity on uniform, finitely-supported processes is an analytic relation.
\end{theorem}
\textit{Uniform} processes come with a canonical enumeration of their transition
relations, including the mass associated to each possible successor state.

By eliminating probabilities altogether but keeping the definability of the
transition relations (thus obtaining \textit{uniformly measurable} LTSs), we obtain:
\begin{theorem}[Theorem~\ref{thm:bisim-MLTS-classifiable-countable-struct}]\label{th:bisim-umltss-class-count}
  Bisimilarity on UMLTSs is classifiable by countable structures.
\end{theorem}
We move on to study bisimilarity on \textit{terminating} processes, that is, those that end
after a finite number of “interactions”.
They correspond to well-founded accessibility relations, where the depth, or rank, is expressed by a countable ordinal.
In the case of bounded ranks, we obtain Borel definability.
\begin{theorem}[Theorem~\ref{thm:bisim_S_leq_alpha_Borel}]\label{th:rank-bounded-segments-wf}
  Bisimilarity on rank-bounded segments of the well-founded part of a UMLTS is Borel.
\end{theorem}
By grouping all countable LTSs into a Polish space, restricting to those of rank $\alpha$, and examining bisimilarity in this setting, we obtain an analogous result:
\begin{theorem}[Theorem~\ref{thm:bisim-wLTSalpha-Borel}]\label{th:bisim-omega-lts}
  Bisimilarity on the family $\omega \LTS^{\leq \alpha}$  of rank-bounded
  (well-founded) countable LTSs is Borel.
\end{theorem}

On the “negative” side, we obtain a lower complexity bound for well-founded
processes of small rank:
\begin{theorem}[Theorem~\ref{thm:E_0-reduction-via-B}]\label{th:bisim-not-smooth}
  Bisimilarity on the family $\WF^{\leq \omega+2}$ of countably branching,
  well-founded trees with rank $\leq\omega+2$ is not smooth.
\end{theorem}
We note a consequence of this result.
\begin{corollary}[Corollary~\ref{cor:ML-not-characterization}]
  There is no countable fragment of basic modal logic with denumerable
  conjunctions that characterizes bisimilarity on $\WF^{\leq \omega+2}$.
\end{corollary}

We briefly describe the contents of the paper.
Section~\ref{sec:rank-restricted-trees} 
gathers some preliminaries on discrete processes, considers
a Polish space of
countable LTSs, and shows that the relation of bisimilarity
$\sim_\alpha$ on well-founded ones of rank ${\leq}\alpha$ is Borel, by using an
appropriate reduction to isomorphism of trees of the same
rank (this will later be used in
Section~\ref{subsect:complexity-bisim-mlts-bounded-rank} to prove
Theorem~\ref{th:rank-bounded-segments-wf}). Section~\ref{sec:e0-reduction}
contains the proof of Theorem~\ref{th:bisim-not-smooth} and its corollary.

Section~\ref{sec:prelims-on-markov} recalls general notions concerning our
Markov decision processes of interest, NLMPs. In Section~\ref{sect:NLMP-substructures} we
determine conditions under which bisimilarity on a process can be inferred from
the behavior of an appropriately defined “substructure”; notable counterexamples
indicate that some care is required; however, for discrete processes, things work out fine.

Pointmass processes are treated in
Section~\ref{sec:pointmass-nlmp}, up to the proof of
Theorem~\ref{th:bisim-unif-finit-supp}. Section~\ref{sec:umlts}
introduces uniformly definable processes, and leverages the results of
Section~\ref{sec:rank-restricted-trees} to obtain
Theorem~\ref{th:bisim-umltss-class-count}.
Section~\ref{sec:conclusion} offers some concluding remarks.

\section{\texorpdfstring{Processes and trees on $\N$}{Processes and trees on ℕ}}
\label{sec:rank-restricted-trees}

For the rest of the paper, $L$ will denote a fixed countable language.
A \emph{labelled transition system} (LTS) with language $L$ is a tuple $\lmp{S}
= (S,\{ R_a \}_{a\in L})$, where $S$ is a set of states and $R_a\sbq S\times S$
are the transition relations. If $s \mathrel{R_a} t$, we write $s \RlabLTS{a}
t$.

We will restrict ourselves to the study of \emph{image-countable} LTSs,
namely, those for which $\{ t \mid s \RlabLTS{a} t \}$ is countable for each $s\in S$
and $a\in L$. Since $L$ is also countable, this leaves us in the realm of \textit{countably branching processes}.

\begin{definition}\label{def:LTS-bisimulation}
  Let $\lmp{S} = (S,\{ R_a \}_{a\in L})$ and 
  $\lmp{S}'= (S',\{ R'_a \}_{a\in L})$ be two LTSs.

  A relation $R \subseteq S \times S'$ is a \emph{bisimulation}
  if $s \mathrel{R} s'$ implies that for every $a \in L$,
  \begin{itemize}
  \item if $s \RlabLTS{a} t$, then there exists $t' \in S'$ 
    such that $s' \RlabLTS{a}' t'$ and $t 
    \mathrel{R} t'$, 
  \item if $s' \RlabLTS{a}' t'$, then 
    there exists $t \in S$ such that $s \RlabLTS{a} t$ and $t 
    \mathrel{R} t'$.
  \end{itemize}
  For states $s\in S$ and $s'\in S'$, we say that $(\lmp{S},s)$ is \emph{bisimilar} to $(\lmp{S}',s')$, denoted $(\lmp{S},s)\sim (\lmp{S}',s')$, if there exists a 
  bisimulation $R$ such that $s \mathrel{R} s'$.
\end{definition}

It is customary to call the first condition \emph{zig} and the second \emph{zag}
(which in other contexts are called “forth” and “back”, respectively). We will
extend this terminology to all bisimulations of the same kind.

The assessment of the relative complexity of bisimilarity is done via Borel reductions: If $E$ and $F$ are two equivalence relations on Polish spaces or standard Borel spaces $X$ and $Y$ respectively, we say that $E$ is Borel reducible to $F$ (in symbols, $E\leq_B F$) if there is a Borel function $f:X\to Y$ such that, for all $x,y \in X$, $x{\, E \,}y \iff f(x){\, F \,}f(y)$.
A reduction $E \leq_B F$ tells us that $E$ is at most as complicated as $F$; once you understand $F$, you understand $E$ (modulo the reduction).

\subsection{Omega expansion and ranks}\label{sec:omega-expansions}

Our primary tool consists of using rooted trees as canonical representatives of
the bisimilarity classes of states of an LTS.
For countably branching processes, reducing bisimilarity to isomorphism requires both a tree unfolding and an appropriate “expansion”. 
The resulting base set will be a tree of sequences over a countable set; such trees are the ones that appear prominently in descriptive set theory.
For a countable set $N$, we denote by $N^{<\N}$ the set of finite sequences over $N$. 
A \emph{tree} $T$ over $N$ is a subset of $N^{<\N}$ closed under initial segments.
The concatenation of sequences $s =(s_i)_{i<n}$ and $t = (t_j)_{j<m}$ is the sequence $s^\smallfrown t = (s_0,\dots, s_{n-1}, t_0, \dots ,t_{m-1})$.
We write $s^\smallfrown a$ for $s^\smallfrown (a)$, if $a \in N$.
Also, we say that $t$ is an extension of $s$ (in symbols, $s \sbq t$) if $t = s^\smallfrown u$, for some $u\in N^{<\N}$.

Following \cite[p.275]{Blackburn:2006:HML:1207696}, we adapt the definition of
expansions to our purposes:
\begin{definition}\label{def:omega-expansion}
	Given an LTS $\lmp{S}=(S,\{R_a\}_{a\in L})$, an 
	\emph{$\omega$-indexed path from $s \in S$} is a 
	sequence of the form
	\[
	u=(s_1,a_1,n_1)(s_2,a_2,n_2)\dots(s_m,a_m,n_m)
	\]
	such that $a_i \in L$, $n_i \in \N$, $s\!\Rlab{a_1}\!s_1$, and 
	for all $i\in\{1,\dots,m-1\}$, 
	$s_i\!\Rlab{a_{i+1}}\!s_{i+1}$.

    For any tree $T$ over $S\times L\times\N$, we define the binary relations $\Suc_T^a\subseteq T\times T$ for $a\in L$ by
        \begin{equation}
          \label{eq:def-Suc}
          (u,v)\in\Suc_T^a  \iff  v=u^\smallfrown(t,a,n) \text{ for some } t\in S \text{ and
          } n \in \N.          
        \end{equation}
        
	The \emph{$\omega$-expansion at $s$} of $\lmp{S}$ is the 
	LTS $(\Tr_{\lmp{S}}(s),\{\Suc_{\lmp{S}}^a\}_{a\in L})$, 
	where $\Tr_{\lmp{S}}(s)$ is the tree over $S\times L\times 
	\N$ of all $\omega$-indexed paths 
	from $s$, and $\Suc_{\lmp{S}}^a \defi \Suc_{\Tr_{\lmp{S}}(s)}^a$.
\end{definition}

The significance of this construction is given by the following characterization:

\begin{theorem}[{\cite[Corollary~47]{BlackburnModelTheory}}]\label{thm:kripke-iso}
  If $(\lmp{S},s)$ and $(\lmp{S}',s')$ are image-countable LTSs, we have
  \[
    (\lmp{S},s)\sim (\lmp{S}',s') \iff
    (\Tr_{\lmp{S}}(s),\{\Suc_\lmp{S}^a\}_{a\in L})\cong
    (\Tr_{\lmp{S}'}(s'),\{\Suc_{\lmp{S}'}^a\}_{a\in L}).
  \]
\end{theorem}
Here, $\cong$ denotes isomorphism of first-order structures.

In the following, we will work with basic modal logic augmented with denumerable
conjunctions and disjunctions, denoted by $\mathrm{ML}_{\omega_1}$ (following
\cite[Sect~5.2 pp.~294--295]{BlackburnModelTheory}). Since LTSs are just Kripke \textit{frames},
atomic propositions will not be used.

We define formulas $\varphi_\alpha$ for $\alpha<\omega_1$ as follows:
\begin{align}\label{eq:rank-formulas}
	\begin{split}
		\varphi_0 & \defi\top, \\
		\varphi_{\alpha+1} & \defi \bigvee_{a\in L}\pos{a} 
		\varphi_\alpha, \\
		\varphi_\lambda & \defi 
		\bigwedge_{\alpha<\lambda}\varphi_\alpha 
		\text{\quad if $\lambda$ is a limit ordinal}.
	\end{split}
\end{align}
Regarding semantics, satisfaction is defined inductively as follows: The formula $\top$ is satisfied by every state, and conjunctions and disjunctions are interpreted in the standard (infinitary) way. 
For the case of modalities, $(\lmp{S},s)\vDash \pos{a}\varphi$ if and only if there exists $t\in S$ such that $s\Rlab{a}t$ and $(\lmp{S},t)\vDash\varphi$.
With these formulas, we can define the rank of a state in an LTS with language $L$.

\begin{definition}\label{def:rank-lts}
	If $\lmp{S}$ is an LTS and $s\in S$, we say that $(\lmp{S},s)$ is \emph{well-founded} if there exists an ordinal $\alpha$ such that $(\lmp{S},s)\nvDash \varphi_{\alpha+1}$. 
	The \emph{rank of $(\lmp{S},s)$} is the smallest of such countable ordinals if it exists, and $\infty$ if it does not. 
	The \emph{well-founded part} of $\lmp{S}$ is the set of states with rank different from $\infty$.
\end{definition}

\begin{prop}\label{prop:bisim-implies-formulas-satisf}
  If $(\lmp{S},s)\sim (\lmp{S}',s')$ then $s$ and $s'$ have the same rank.
\end{prop}
\begin{proof}
  By the invariance of $\mathrm{ML}_{\omega_1}$ with respect to bisimilarity of
  image-countable LTS.
\end{proof}

Going back to $\omega$-expansions, well-founded sequence trees (those which do not have infinite branches) also have their
own notion of rank.
\begin{definition}\label{def:rank-tree}
	If $T$ is a well-founded tree on $A$, the \emph{rank of $s \in
        A^{<\mathbb{N}}$} is recursively defined as $\rho_T(s) \defi \sup\{\rho_T(t)+1 \mid t \in T, \; s \subsetneq t\}$.
	Furthermore, $\rho_T(s) = 0$ if $s$ has no extensions in $T$ or if $s \notin T$. 
	The \emph{rank of $T$} is defined as $\rho(T) \defi \sup\{\rho_T(s)+1 \mid s \in T\}$.
\end{definition}
We will next show that both notions of rank actually coincide for trees. Indeed,
Definition~\ref{def:rank-lts} could have been stated to refer to the tree rank
of the $\omega$-expansion, but we preferred to use a logical description since
we consider this approach closer to the original LTS, and it is amenable to being
generalized to settings where one cannot find a bisimilar tree.

We will use $\WF_N$ to denote the space of well-founded trees over a countable set $N$.
Additionally, $\WF^{\bowtie \alpha}_N$ (with ${\bowtie} \in \{ =, <, \leq, >, \geq \}$)
will denote the subfamily consisting of trees with rank ``$\bowtie\alpha$.''

Trees can be considered to be LTSs over a singleton set of labels
$\{ \star \}$, where the transition relation is given by $s\raisebox{0.5ex}[0.1ex]{\Rlab{\star}} t$ if
and only if $t$ is an immediate successor of $s$ in $T$. As an auxiliary
step for relating well-founded LTSs to trees with the same ranks, we define
analogous formulas corresponding to $\varphi_\alpha$ from
\eqref{eq:rank-formulas} when taking $L=\{ \star \}$. To write them down, and in
the following, we use $\bigd$ instead of $\pos{\star}$.
\begin{align*}
	\psi_0 & \defi \top, \\
	\psi_{\alpha+1} & \defi \bigd \psi_\alpha, \\
	\psi_\lambda & \defi 
	\bigwedge_{\alpha < \lambda} \psi_\alpha 
	\text{\quad if $\lambda$ is a limit ordinal}.
\end{align*}

In the next easy result, note that $\Tr_{\lmp{S}}(s)$ is just the tree (with
the single transition relation corresponding to the $\star$ symbol) and should not be conflated
with the $\omega$-expansion using the full $L$.
\begin{lemma}\label{lem:formula-correspondence-wLTS}
	$(\lmp{S}, s) \vDash \varphi_\alpha \iff   
	\bigl((\Tr_{\lmp{S}}(s), \{{\stackrel{\star}{\to}}\}) , \emptyset\bigr) \vDash \psi_\alpha$.
        \qed
\end{lemma}

We now give a characterization of $\WF^{\leq\alpha}_N$ in terms of the satisfaction of the formulas $\psi_\alpha$ for $\alpha < \omega_1$.
An induction argument proves the next proposition (for details, see Appendix~\ref{sec:missing_proofs}). This, along with the previous lemma, gives the desired rank preservation property.

\begin{prop}\label{prop:WF-tree-formulas}
	$T \notin \WF^{<\alpha}_N \iff (T, \emptyset) \vDash \psi_\alpha$. \qed
\end{prop}
 
We end this section by recalling that $\WF_N^{\leq\alpha}$ is Borel definable, which is required for the reductions to be proved in Sections~\ref{sec:omega-lts} and \ref{subsect:complexity-bisim-mlts-bounded-rank}.
	
\begin{prop}
	\label{prop:WF-alpha-Borel-standard}
	$\WF_N^{\leq \alpha}$ is a standard Borel space.
\end{prop}
\begin{proof}
	By \cite[Exr.~34.6]{Kechris}, the map $T \mapsto \rho(T)$ is a \textit{$\mathbf{\Pi}_1^1$-rank} on the 
	set $\WF_\N\cong \WF_N$, which in particular means that the initial segments $\WF_N^{\leq
          \alpha} =  \{ x \in \WF_N \mid \rho(x) \leq \alpha \}$ belong in $\Borel(\Tr_N)$.
\end{proof}

\subsection{Omega LTS}
\label{sec:omega-lts}

Next, we will consider the Polish space of all pointed countable LTSs and their corresponding notion of bisimilarity.
As usual, we identify the powerset of a (countable) set $N$ with the Cantor
space $2^N$.
\begin{definition}\label{def:omegaLTS}
  \begin{enumerate}
  \item
    We denote by $\omega \LTS$ the product space 
    $\N\times \prod_{a\in L} 2^{\N\times \N}$ (with $\N$ 
    and $2\defi \{ 0, 1 \}$ both discrete). 
  \item
    For each $x\in \omega \LTS$, noted as $x=(x_*,(x_a)_{a\in L})$, let $\ltsfrom_x$ be
    the pointed LTS over $\N$ with root $x_*$ and transition relations $R_a$ 
    satisfying $\chi_{R_a}=x_a$ $(a\in L)$ (where $\chi$ denotes the characteristic function).
  \item
     For each $x,y \in \omega \LTS$, we write $x\sim y$ if and only if $\ltsfrom_x \sim \ltsfrom_y$.
  \end{enumerate}
\end{definition}

\begin{prop}\label{prop:semantic-is-Borel-wLTS}
  For every $\varphi\in\mathrm{ML}_{\omega_1}$, $\{ x \in \omega\LTS \mid \ltsfrom_x
    \vDash \varphi \}$ is a Borel subset of $\omega \LTS$.
\end{prop}
\begin{proof}
	Throughout this proof we will denote $\{ x \in \omega\LTS \mid \ltsfrom_x
	\vDash \varphi \}$ as $\sem{\varphi}$, and proceed by structural induction.
	The base case $\varphi = \top$ gives $\sem{\top}=\omega \LTS$, and the cases for the logical connectives $\neg$, $\bigwedge$, and $\bigvee$ follow easily from the IH.
	
	For the modal formulas $\pos{a}\varphi$ we will use the following auxiliary functions: for each $t\in \N$, define $g_t:\omega LTS \to \omega \LTS$ by $g_t(i,(x_a)_{a\in L}) = (t,(x_a)_{a\in L})$, that is, $g_t$ moves the root of the associated LTS to $t$.
	Each $g_t$ is continuous since, for $a\in L$ and $i,j\in 
	\N$, $g_t^{-1}[\{x\in \omega \LTS\mid x_* = s \wedge x_a(i,j) = 1\}]$ equals $\{x\in \omega \LTS\mid x_a(i,j) = 1\}$ if $s=t$, and $\emptyset$ otherwise.
	We have
	\begin{align*}
		\sem{\pos{a}\varphi} &= \{(x_*, (x_a)_{a\in L})\in \omega\LTS \mid \exists t, x_a(x_*,t) = 1 \wedge (t, (x_a)_{a\in L}) \in \sem{\varphi}\}  \\
		&= \{(x_*, (x_a)_{a\in L}) \in \omega\LTS \mid \exists t, x_a(x_*,t)=1  \wedge g_t(x_*, (x_a)_{a\in L}) \in \sem{\varphi}\} \\
		&= \{(x_*, (x_a)_{a\in L}) \in \omega\LTS \mid \exists t, x_a(x_*,t)=1  \wedge (x_*, (x_a)_{a\in L}) \in g_t^{-1}(\sem{\varphi})\} \\	 
		&= \bigcup_{t\in \N} \{(x_*, (x_a)_{a\in L}) \in \omega\LTS \mid x_a(x_*,t)=1 \} \cap g_t^{-1}(\sem{\varphi}).		
	\end{align*}
	This set is Borel by measurability of the projections $p_a:\omega \LTS \to \N$, the continuity of $g$ and the IH.
\end{proof}

\begin{definition}\label{def:omegaLTS-alpha}
  If $\alpha < \omega_1$, we denote by $\omega \LTS^{\leq \alpha}$
  the set of all $x\in\omega\LTS$ such that $\ltsfrom_x$ is a well-founded
  pointed LTSs with rank at most $\alpha$.
\end{definition}

\begin{corollary}\label{cor:omegaLTS_leq_alpha_borel}
	For each $\alpha < \omega_1$, the set $\omega \LTS^{\leq \alpha}$ is Borel in $\omega \LTS$.
\end{corollary}
\begin{proof}
	$x \in \omega \LTS^{\leq \alpha}\iff \ltsfrom_x\nvDash 
	\varphi_{\alpha+1}\iff x\notin \{ x \in \omega\LTS \mid \ltsfrom_x
	\vDash \varphi_{\alpha+1} \}$.
\end{proof}

We can also calculate the $\omega$-expansion of elements in $\omega \LTS$:
\begin{definition}\label{def:Omega}
  Let $M \defi \N \times L \times \N$.
  The map
  \[
    \Omega: \omega \LTS \to \Tr_M \times \prod_{a \in L} 2^{M^{<\N}
      \times M^{<\N}}
  \]
  is defined by $\Omega(x) \defi (\Tr_{\lmp{S}}(x_*),\{\Suc_\lmp{S}^a\}_{a\in
    L})$, where $\lmp{S} = \ltsfrom_x$.
\end{definition}
When necessary, we will write $\Omega$ in components as in
\[
  \Omega(x) =
  (\Omega_0(x), \{ \Omega_a(x) \mid a \in L \}).
\]
We have this consequence of Theorem~\ref{thm:kripke-iso}:
\begin{lemma}\label{lem:sim-iff-Omega-cong}
  For all $x, y \in 
  \omega \LTS$, $x \sim y \iff \Omega(x) \cong  \Omega(y)$. \qed
\end{lemma}
Recall from (\ref{eq:def-Suc}) in Definition~\ref{def:omega-expansion} that the $\Suc_\lmp{S}^a$ relations
are actually definable from the tree $\Omega_0(x)$. Hence bisimilarity is
reducible to the following relation $\equiv$ on $\Tr_M$:
\begin{equation}
  \label{eq:equiv-def}
  T \equiv T' \iff (T, \{\Suc^a_T\}_{a \in L}) \cong (T', \{\Suc^a_{T'}\}_{a \in
    L}).
\end{equation}

\begin{lemma}\label{lem:sim-iff-Omega0-equiv}
  For all $x, y \in 
  \omega \LTS$, $x \sim y \iff \Omega_0(x) \equiv  \Omega_0(y)$. \qed
\end{lemma}
\begin{lemma}\label{lem:Omega-continuous}
	$\Omega$ is continuous.
\end{lemma}
\begin{proof}
	We show that the component functions of
	$\Omega(x)=(\Omega_0(x),\{\Omega_a(x)\mid a \in L\})$ are continuous.
	Fix	$u=(n_0,a_0,m_0)\dots(n_k,a_k,m_k)\in M^{<\N}$.  We compute
	\begin{align*}
		\Omega_0^{-1}[\{T\in \Tr_M &\mid u\in T\}] = \\
		&= \{x\in \omega\LTS \mid u\in \Omega_0(x)\} \\
		&\textstyle
		= \{x \mid x_{a_0}(x_*,n_0)=1
		\wedge \bigwedge_{1\le i\le k} x_{a_i}(n_{i-1},n_i)=1\} \\
		&\textstyle
		= \{x \mid x_{a_0}(x_*,n_0)=1\}
		\cap \bigcap_{1\le i\le k}\{x \mid x_{a_i}(n_{i-1},n_i)=1\}.
	\end{align*}
	The first set in this intersection equals
	\[
	\bigcup_{b\in \N}
	p_{a_0}^{-1}\bigl[\{f\in 2^{\N\times\N}\mid f(b,n_0)=1\}\bigr]
	\cap p_*^{-1}[\{b\}],
	\]
	which is open in $\omega\LTS$.  Likewise, the set
	$\{x \mid x_{a_i}(n_{i-1},n_i)=1\}$ is open for each
	$i\in\{1,\dots,k\}$.
	
	For $\Omega_a$ we must show that, for $a\in L$ and
	$u,v\in M^{<\N}$, the set
	\[	
	\Omega_a^{-1}\bigl[\{X\subseteq M^{<\N}\times M^{<\N} \mid
	(u,v)\in X\}\bigr]
	= \{x\in \omega\LTS \mid (u,v)\in \Omega_a(x)\}
	\]
	is open.  
	If $v=u^\smallfrown(n,a,m)$ for some $n,m\in\N$, this set is
	$\{x \mid v\in \Omega_0(x)\}$, which is open by continuity of $\Omega_0$.  
	If $u,v$ are not of this form, the set is empty.
\end{proof}

\begin{lemma}\label{lem:Omega0-codomain}
	The image of the restriction of $\Omega_0$ to 
	$\omega \LTS^{\leq \alpha}$ is contained in $\WF_M^{\leq \alpha}$.
\end{lemma}
\begin{proof}
	For all $x\in \omega LTS$, we can apply Lemma~\ref{lem:formula-correspondence-wLTS} and Proposition~\ref{prop:WF-tree-formulas} to get $\ltsfrom_x \vDash \varphi_\alpha \iff (\Tr_{\ltsfrom_x}(x_*), \emptyset) \vDash \psi_\alpha \iff \Omega_0(x) \notin \WF_M^{<\alpha}$.
\end{proof}

Summing up all of the previous results, we can show that the restriction of
bisimilarity on \( \omega \LTS \) to \( \omega \LTS^{\leq \alpha} \) is Borel.

\begin{theorem}\label{thm:bisim-wLTSalpha-Borel}
	For each \( 1 \leq \alpha < \omega_1 \), \( \sim \) on \( \omega \LTS^{\leq \alpha} \) is Borel.
\end{theorem}
\begin{proof}
  By Lemma~\ref{lem:Omega0-codomain}, the codomain of $\Omega_0$, when restricted
  to  $\omega \LTS^{\leq \alpha}$,
  is the space
  $\WF_M^{\leq \alpha}$; these two spaces are standard Borel by
  Corollary~\ref{cor:omegaLTS_leq_alpha_borel} and
  Proposition~\ref{prop:WF-alpha-Borel-standard}, respectively. From
  Lemma~\ref{lem:sim-iff-Omega0-equiv} and Lemma~\ref{lem:Omega-continuous} we know
  that $\Omega_0$ is a continuous reduction witnessing ${\sim}\rest(\omega \LTS^{\leq
    \alpha}) \leq_B {\equiv}$.  It can be shown \cite[Section
    (1.2)]{friedman_stanley_1989} that ${\equiv}\rest(\WF_M^{\leq \alpha})$ is
  Borel (for a self-contained proof, see
  Appendix~\ref{sec:isom-rank-restr-trees}).  Consequently, ${\sim}\rest(\omega
  \LTS^{\leq \alpha})$ also is.
\end{proof}

\subsection{\texorpdfstring{%
  $E_0$ reduction and an application to $\mathrm{ML}_{\omega_1}$}
  {E₀ reduction and an application to MLω₁}
}
\label{sec:e0-reduction}
We will provide a Borel reduction of $E_0$, the  relation of eventual equality of
binary sequences given by $x \mathrel{E_0} y \iff \exists n \forall m \geq n \;
(x(m) = y(m))$,  to bisimilarity between well-founded trees of small rank.
It is convenient to think of $E_0$ as a relation on subsets of $\N$, identifying each \( x \subseteq \N \) with its characteristic function \( \N \to 2 \).

Recall that an equivalence relation is called \emph{smooth} if it is Borel
reducible to the identity relation on $\R$, or equivalently, if there exists a countable family $\Fam$ of Borel sets such that $x \mathrel{R} y \iff \forall A\in \Fam \; (x\in A \iff y\in A)$.
We will denote the equivalence relation induced by the right-hand side by $\Rel(\Fam)$.
The relation $E_0$ occupies a special place in the descriptive set-theoretic hierarchy of equivalence relations: among non-smooth Borel equivalence relations, it has minimum complexity (up to Borel bi-reducibility). 
Indeed, the Harrington–Kechris–Louveau dichotomy theorem \cite[Thm.1.1]{Harrington1990AGD} states that if $E$ is a Borel equivalence relation on a standard Borel space, then either $E$ is smooth or there is a continuous embedding of $E_0$ into $E$.

We fix a tree 
\[
  A(\N) \defi \{ \emptyset \} \cup \{ (n, 0, \dots, 0) \mid \text{at most
    }n\text{ zeros after the first term} \}
\]
on \( \N \) that has a unique branch of
length \( k+1 \) for each \( 0 \leq k \in \N \), so that branches corresponding to different lengths are incompatible, i.e., they have no segments in common (see Figure~\ref{fig:T_omega-T_even}, left).
We say that \( u \in \N^{<\N} \) is \emph{admissible} if \( u \in A(\N) \).
We note that for each admissible \( u \), its first term $u_0$ determines the branch to which it belongs.
Given a subset \( x \subseteq \N \), we consider the subtree \( A(x) \)
of \( A(\N) \) consisting of its branches of length \( k+1 \) for each \( k
\in x \) (see Figure~\ref{fig:T_omega-T_even}, right).

\definecolor{myorange}{RGB}{0, 0, 0}
\definecolor{mygreen}{RGB}{200, 200, 200}
\definecolor{darkgray}{RGB}{0, 0, 0}

\tikzset{
  tree label/.style={font=\small, align=center, anchor=south},
  caption label/.style={font=\normalsize, text=black},
  node style/.style={font=\footnotesize, inner sep=2pt, text=black},
}
	
\begin{figure}[h]
  \begin{center}
    \begin{forest}
      for tree={
        node style,
        edge={draw=darkgray, thick},
        l sep=20pt, 
        s sep=10pt,
      }
      [$\emptyset$, name=rootA, 
        [(0)]
        [(1) [{(1,0)}] ] 
        [(2) [{(2,0)} [{(2,0,0)}] ] ]
        [(3) [{(3,0)} [{(3,0,0)} [{(3,0,0,0)} [{}, edge={draw=none} ] ] ] ] ]
        [\ldots] 
      ]
    \end{forest}
    \qquad
    \begin{forest}
      for tree={
        node style,
        edge={draw=darkgray, thick},
        l sep=20pt, 
        s sep=10pt,
      }
      [$\emptyset$, name=rootB,
        [(0)]
        [(2) [{(2,0)} [{(2,0,0)}] ] ]
        [(4) [{(4,0)} [{(4,0,0)} [{(4,0,0,0)} [{(4,0,0,0,0)}  ] ] ] ] ]
        [\dots] 
      ]
    \end{forest}
  \end{center}
  \caption{The trees $A(\N)$ and $A(\{ 0, 2, 4, \dots \})$.}\label{fig:T_omega-T_even}
\end{figure}
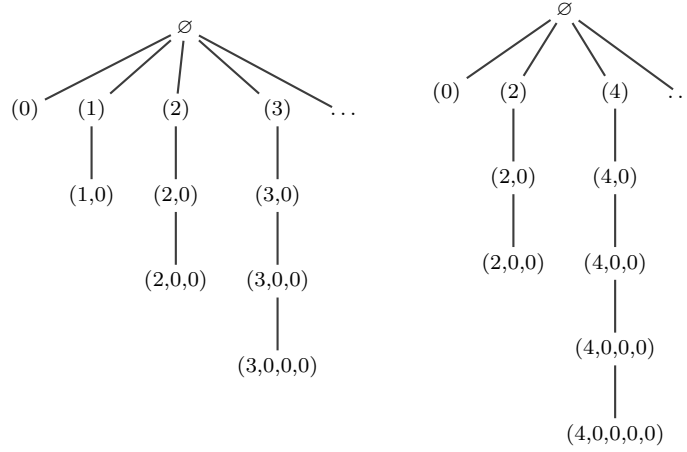

\begin{lemma}\label{lem:T_x-formula}
  If \( x \subseteq \N \), then \( (A(x), \emptyset) \vDash
\bigd^{k+1}(\neg \bigd \top) \iff k \in x \). \qed
\end{lemma}

For measurability purposes, we need a uniform and continuous way of enumerating all finite modifications of a point $x=(x_0,x_1,\ldots)\in 2^{\mathbb{N}}$. 
More precisely, we seek a family of continuous functions $m_n:2^{\mathbb{N}}\to 2^{\mathbb{N}}$ such that $y\in\{m_n(x)\mid n\in\mathbb{N}\}$ if and only if $x\mathrel{E_0}y$.
To construct such a family, we use the binary expansion of each $n\in\mathbb{N}$, viewed as an element $b_n\in 2^{\mathbb{N}}$, and define $m_n(x)=b_n \mathrel{\mathrm{XOR}} x$, where $\mathrm{XOR}$ denotes coordinatewise exclusive disjunction. 
These functions satisfy the required properties. 
Note also that $m_0$ is the identity function.

Given \( x \subseteq \N \), we construct a tree \( B(x) \) on \(
\N \): 
We prefix with $n$ each node of the tree $A(m_n(x))$ corresponding to the $n$th modification of
$x$. Then we adjoin $\emptyset$ as the new root (see Figure~\ref{fig:B_x}).
We observe that the rank of \( B(x) \) is at most \( \omega + 2 \) (it equals
$\omega + 1$ if and only if $x$ is finite).
Moreover, \( B(x) \in 2^{\N^{<\N}} \) forms an LTS with the successor relation.
\begin{figure}[h]
  \begin{center}
    \begin{forest}
      for tree={
        node style,
        edge={draw=darkgray, thick},
        l sep=20pt, 
        s sep=10pt,
      }
      [$\emptyset$, name=rootC, s sep=15pt,
        [(0), text=myorange, edge={draw=darkgray}, name=startOrange
          [{(0,0)}, text=myorange, edge={draw=myorange}, name=leftOrange]
          [{(0,2)}, text=myorange, edge={draw=myorange} 
	    [{(0,2,0)}, text=myorange, edge={draw=myorange} 
	      [{(0,2,0,0)}, text=myorange, edge={draw=myorange}] 
	    ] 
          ]
          [{(0,4)\ \dots}, text=myorange, edge={draw=myorange} 
	    [{(0,4,0)}, text=myorange, edge={draw=myorange} 
	      [{(0,4,0,0)}, text=myorange, edge={draw=myorange}, name=midOrange
	        [{(0,4,0,0,0)}, text=myorange, edge={draw=myorange}
	          [{(0,4,0,0,0,0)}, text=myorange, edge={draw=myorange}, name=endOrange]
	        ]
	      ]
	    ] 
          ]
        ]
        [(1) [$A(m_1(x))$] ]
        [(2) [$A(m_2(x))$] ]
        [(3) [$A(m_3(x))$] ]
        [\dots, edge={draw=none}] 
      ]
      ]
        \draw[mygreen, thick] 
        (startOrange.west) 
        to[out=180, in=90] ([xshift=-1.2cm]leftOrange.east) 
        to[out=270, in=180] ([yshift=-0.5cm]endOrange.south)
        to[out=0, in=270] ([xshift=0.8cm]endOrange.east)
        to[out=90, in=0] (startOrange.east);
    \end{forest}
  \end{center}
  \caption{The tree $B(x)$ with explicit $A(m_0(x))$ for $x = \{ 0, 2, 4, \dots \}$.}\label{fig:B_x}
\end{figure}
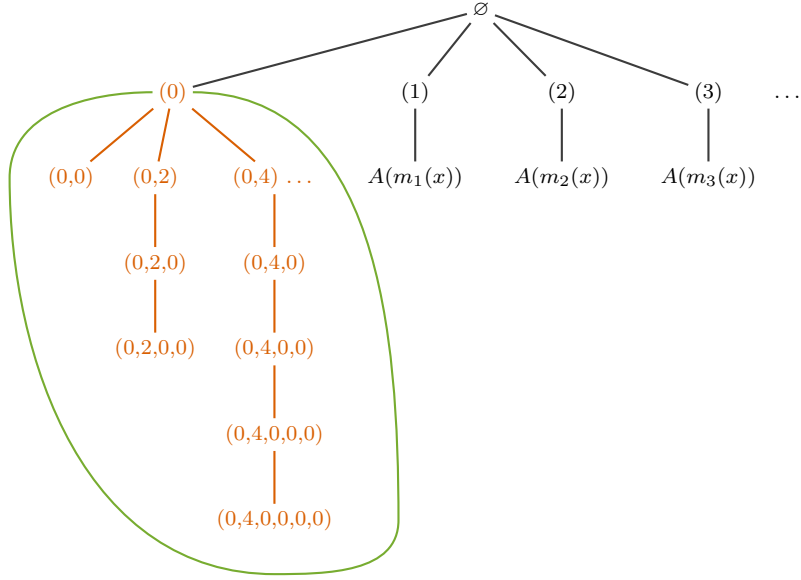

\begin{lemma}\label{lem:S-continuous}
	$B:2^\N\to \WF_{\N}^{\leq \omega+2}$ is a continuous injection.
\end{lemma}
\begin{proof}
  We divide the construction of the function $B$ into two steps:
  \begin{enumerate}
    \item Let $M:2^\N\to (\Tr_{\N})^\N$ be given by 
    $M(x)=(A(m_i(x)))_{i\in \N}$. We show that $M$ is continuous; let $u\in \N^{<\N}$:
    \begin{align*}
      (\pi_n\circ M)^{-1}&[\{T\in \Tr_{\N}\mid u\in 
            T\}]=\\
      &=\{x\in 2^\N\mid u\in \pi_n(M(x))\} \\
      &= \{x\in 2^\N \mid u\in A(m_n(x))\} \\
      &=
      \begin{cases*}
        \emptyset, & if $u$ is not admissible,\\
        \{x\in 2^\N \mid m_n(x)(u_0)=1\}, & if $u$ is admissible.
      \end{cases*}
    \end{align*}
    The continuity of $m_n$ ensures that the second branch is an open set.
    
    \item Now let $P:(\Tr_{\N})^\N \to 
    \Tr_{\N}$ be the “union” map that glues a 
    sequence of trees to a new root:
    \[
      P\bigl((T_i)_{i\in \N}\bigr)
      =\{\emptyset\}\cup \{(i)^\smallfrown s\mid 
      s\in T_i\}.
    \]
    We verify that $P$ is also continuous. Let $u\in 
    \N^{<\N}$ be of the form $(k)^\smallfrown v$ 
    for some $k\in \N$ and $v\in \N^{<\N}$.
    We compute:
    \begin{align*}
      P^{-1}[\{T\in \Tr_{\N} \mid u\in T\}]
      &=\bigl\{(T_i)_i\in 
      \textstyle\prod_{n\in \N}\Tr_{\N}\mid v\in T_k\bigr\} \\
      &=\textstyle\prod_{i<k}\Tr_{\N}\times \{T\in 
      \Tr_{\N}\mid v\in T\}\times \textstyle\prod_{i>k}\Tr_{\N}.
    \end{align*}
    This set is open in the product topology. 
  \end{enumerate}
  Since $B=P\circ M$, we have proved that it is continuous. $B$ is seen to be
  injective by observing that $x\mapsto A(m_0(x))$ is.  
\end{proof}

\begin{theorem}\label{thm:E_0-reduction-via-B}
	If $x,y\subseteq \N$, then $x \mathrel{E_0} y \iff 
	(B(x),\emptyset)\sim (B(y),\emptyset)$. 
	Consequently, $E_0 \leq_B {\sim}_{\WF^{\leq \omega+2}}$.
\end{theorem}
\begin{proof}
  $(\Rightarrow)$ If $x$ and $y$ 
  differ on finitely many elements, the set of their finite 
  modifications is the same and there exists a bijection 
  $h:\N\to \N$ such that $m_n(x)=m_{h(n)}(y)$. 
  The relation $R\subseteq B(x)\times B(y)$ that glues the roots and 
  the corresponding representative subtrees of $m_n(x)$ and 
  $m_{h(n)}(y)$ is an isomorphism and, a fortiori, a bisimulation. 
  
  $(\Leftarrow)$ Given $z\subseteq \N$, consider the formula
  \[
    \varphi(z)\defi 
    \bigwedge_{n\in z}\bigd^{n+1}(\neg \bigd\top)\wedge 
    \bigwedge_{n\notin z} \neg(\bigd^{n+1}(\neg \bigd\top)).
  \]
  By Lemma~\ref{lem:T_x-formula} we observe that 
  $(A(w),\emptyset)\vDash \varphi(z) \iff w=z$.
  
  Suppose that $x\mathrel{\cancel{E_0}} y$. 
  By construction, since $m_0(x)=x$, we have 
  $(B(x),\emptyset)\vDash \bigd \varphi(x)$. 
  However, $(B(y), \emptyset)$ does not satisfy this formula, 
  since none of its subtrees $A(m_n(y))$ satisfies 
  $\varphi(x)$.
  A contradiction because $(B(x),\emptyset)$ and $(B(y), \emptyset)$ satisfy the same formulas.
\end{proof}

If we work with an LTS on a measurable space, we say that a logic is \emph{measurable} if the validity sets of its formulas are measurable sets. 
By the Harrington-Kechris-Louveau dichotomy, we deduce that $\sim$ is not smooth and obtain the following corollary.

\begin{corollary}
	There is no countable and measurable logic that characterizes bisimilarity for well-founded trees with countable branching and rank $\leq\omega+2$.
\end{corollary}

We highlight that there does exist at least one countable logic that
characterizes bisimilarity for countable LTS
\cite[Exm.13]{2012arXiv1211.0967S}, but as stated above, no such logic can be measurable.

Since $\mathrm{ML}_{\omega_1}$ has measurable validity sets, we can conclude:
\begin{corollary}\label{cor:ML-not-characterization}
	There is no countable fragment of $\mathrm{ML}_{\omega_1}$ that characterizes bisimilarity for well-founded trees with countable branching and rank $\leq\omega+2$.
\end{corollary}

\section{Preliminaries on Markov processes}
\label{sec:prelims-on-markov}
We now turn to processes that involve stochastic behavior; we review some basic
concepts and set notation, mostly borrowed from \cite{MOSS2004279,coco,D'Argenio:2012:BNL:2139682.2139685}.

The set of subprobability measures over a measurable space
$(S,\Sigma)$ will be denoted by $\Delta(S)$.
If $R$ is a relation on $S$, its \emph{lifting} $\bar{R}$ to $\Delta(S)$ is the equivalence relation given by:
\begin{equation}\label{eq:int-lift}
	\mu \mathrel{\bar{R}} \mu' \iff \forall Q\in \Sigma(R) \;
	\mu(Q)=\mu'(Q),
\end{equation}
where $\Sigma(R)$ stands for the sub-$\sigma$-algebra of measurable $R$-closed sets: $E\in \Sigma$ such that $\{s\in S \mid \exists x\in E \; (x\mathrel{R}s \vee s\mathrel{R}x)\} \subseteq E$.
It is useful to think that $\Sigma(R)$ contains all the information about $R$ that can be measurably encoded.

Let $\Delta(\Sigma)$ be the smallest $\sigma$-algebra on $\Delta(S)$ making every evaluation function $\mathrm{ev}_E:\Delta(S) \to [0,1]$, $E\in \Sigma$, measurable with respect to $\Borel([0, 1])$.
If $E\in \Sigma$ and ${\bowtie} \in \{<,\leq, 
>,\geq\}$, $\Delta^{\bowtie q}(E)$ denotes the set
$\{\mu\in \Delta(S) \mid \mu(E)\bowtie q\}$ (somewhat more usually called
  $\beta^{\bowtie q}(E)$---a new terminology comes handy in highlighting that we
  are not only encompassing probabilities).
We recall that the family of all such
sets for ${\bowtie} = {<}$, $q\in \Q\cap [0,1]$ and
$E\in \Sigma$ generates $\Delta(\Sigma)$.
In general, if $\Gamma \sbq \Sigma$, $\Delta(\Gamma)\defi
\sigma(\{\Delta^{<q}(E)\mid E\in \Gamma\})$.

The following definition, essentially taken from \cite[Def.~3.2]{D'Argenio:2012:BNL:2139682.2139685},
was later baptized in \cite{ROCKS} (together with some accompanying bisimulation
notions).
\begin{definition}\label{def:hit-algebra}
	Given a set $X$ and a family $\Gamma\sbq \Power(X)$, the \emph{hit
		$\sigma$-algebra} $H(\Gamma)$ is the smallest
	$\sigma$-algebra on $\Gamma$ that contains all sets $H_D\defi \{G\in
	\Gamma\mid G\cap D \neq\emptyset\}$ with $D \in \Gamma$.
\end{definition}

\begin{definition}\label{def:nlmp}
	A \emph{nondeterministic labelled Markov process}, or NLMP, is a structure $\lmp{S} 
	=(S,\Sigma,\{\Tfont{T}_a\mid a\in L\})$ where $\Sigma$ is 
	a $\sigma$-algebra over the state set $S$, and 
	for each label $a\in L$, $\Tfont{T}_a:(S,\Sigma)\to 
	(\Delta(\Sigma),H(\Delta(\Sigma)))$ is measurable. 
	We say that $\lmp{S}$ is \emph{image-finite (countable)} if all sets $\Tfont{T}_a(s)$ are finite (countable). 
	NLMPs that are not image-countable are called \emph{image-uncountable}.
\end{definition}
NLMP were introduced in \cite{DWTC09:qest}, as a common generalization of
\begin{itemize}
\item
  \emph{probabilistic automata} \cite{SegalaL94}, to the continuous setting; and
\item
  \emph{labelled Markov processes (LMP)} \cite{Des98LogicalChar,doi:10.1142/p595}
  and \emph{stochastic relations} \cite{doberkat2007stochastic}, with internal nondeterminism.
\end{itemize}
It is to
be noted that there are several other related notions in the literature; to cite
some, recall \emph{infLMP} \cite{DBLP:journals/iandc/DesharnaisLT11} and even
higher processes with both angelic and demonic components, as in
\emph{stochastic effectivity functions} \cite{2014arXiv1405.7141D}.

The following notion is a straightforward common generalization of the
corresponding LMP concept to a nondeterministic setting. 
\begin{definition}\label{def:nlmp_sb}
  A relation $R \sbq S\times S$ is a \emph{state bisimulation} on an NLMP
  $(S,\Sigma,\{\Tfont{T}_a\mid a\in L\})$ if it is symmetric and for all $a \in
    L$, $s\mathrel{R}t$ implies that for every $\mu \in \Tfont{T}_a(s)$ there
    exists $\mu'\in \Tfont{T}_a(t)$ such that $\mu \mathrel{\bar{R}}\mu'$.
    
    We say that $s,t\in S$ are state \emph{bisimilar}, denoted by $s\mathrel{\sim_\sfs}t$, if there is a state bisimulation $R$ such that $s\mathrel{R}t$.
\end{definition}

We now turn to the \emph{external} version of state bisimulation. This was
championed by Bacci in \cite{bacci,Bacci2014BisimulationOM}, from which we take the
following concepts.
\begin{definition}\label{def:closed-pair}
	Let $R\subseteq S\times S'$ be a relation, and let $A\subseteq S$ and $A'\subseteq S'$. 
	The pair $(A,A')$ is called \emph{$R$-closed pair} if $R\cap(A\times S') = R\cap(S\times A')$.
\end{definition}

The external view of bisimulation was further analized in a previous work, from
which we extract some auxiliary results.
\begin{lemma}[{\cite[Lemma~3.20]{survey_bisim}}]\label{lem:pair-equiv}
	Let $R\sbq S\times S'$, $A\sbq S$, and $A'\sbq S'$.
	\begin{enumerate}
		\item\label{item:in-equiv} $(A,A')$ is an $R$-closed pair iff $R[A]\subseteq A'$ and $R^{-1}[A']\subseteq A$.
		\item\label{item:closure-prop} The family of $R$-closed pairs is closed under complementation
		and arbitrary unions and intersections in coordinates.
		\item\label{item:closed-pairs-antimonotonicity} If $R_0\subseteq
		R_1$, then every $R_1$-closed pair is also $R_0$-closed.
	\end{enumerate}
\end{lemma}

Given $R \subseteq S \times
S'$, we use the notation $\Sigma^\times(R)$ for the
family of measurable $R$-closed pairs $(Q, Q')$. 
The lift $\bar{R}$ of $R$ to $\Delta(S) \times \Delta(S')$ is defined as
\begin{equation}\label{eq:ext-lift}
	\mu \mathrel{\bar{R}} \mu' \iff \forall (Q, Q') \in
	\Sigma^{\times}(R) \; \mu(Q) = \mu'(Q').
\end{equation}

\begin{definition}\label{def:nlmp-ext-sb}
  A relation $R \subseteq S\times S'$ is
  an \emph{external state bisimulation} if for every $a \in L$,
  $s\mathrel{R}s'$ implies that for every $\mu \in \Tfont{T}_a(s)$
  there exists $\mu'\in \Tfont{T}'_a(s')$ such that $\mu
  \mathrel{\bar{R}}\mu'$ (“zig”) and for every $\mu' \in \Tfont{T'}_a(s')$
  there exists $\mu\in \Tfont{T}_a(s)$ such that $\mu
  \mathrel{\bar{R}}\mu'$ (“zag”).

  We will say that $s \in S, s'\in S'$ are \emph{externally bisimilar}, denoted
  by $s\sim_\sfs^\times s'$, if there exists an external state
  bisimulation $R$ such that $s\mathrel{R}s'$.
\end{definition}

As usual, it can be proved that the union of an arbitrary family of external state bisimulations is an external state bisimulation. Therefore, $\sim_\sfs^\times$ is an external state bisimulation.

\begin{lemma}[{\cite[Lemma~3.21]{survey_bisim}}]\label{lem:closed-vs-closed-pair}
	If $R\subseteq S\times S$, the following statements hold:
	\begin{enumerate}
		\item\label{item:E-closed-iff-EE-closed} $E$ is $R$-closed if and only if $(E,E)$ is an $R$-closed pair.
		\item\label{item:closed-pair-reflexive-case} If $R$ is reflexive and $(E,E')$ is an $R$-closed pair, then $E=E'$.
	\end{enumerate}
\end{lemma}

The first item in this lemma implies that, whenever $\lmp{S} = \lmp{S'}$, every external symmetric state bisimulation $R$ is an internal state bisimulation. 
When $R$ is reflexive, we get an equivalence $(Q, Q') \in \Sigma^\times(R) \iff Q = Q' \in \Sigma(R)$. Therefore, the liftings	\eqref{eq:int-lift} and \eqref{eq:ext-lift} essentially coincide (modulo the correspondence $x\longleftrightarrow (x,x)$) and we get the following result.

\begin{lemma}[{\cite[Lemma~4.16]{survey_bisim}}]\label{lem:S-equals-S'-case}
	If $\lmp{S}=\lmp{S}'$, the external definition of state
	bisimulation coincides with the internal one in
	the reflexive \textup{[NB: and symmetric]} case.
\end{lemma}

\section{Substructures}\label{sect:NLMP-substructures}

Given a fixed NLMP $\lmp{S}=(S,\Sigma,\{\Tfont{T}_a\mid a\in L \})$, we are
interested in a notion of substructure of $\lmp{S}$. Recall that for $A\sbq S$, we have the
following family of its subsets, $\Sigma\rest A=\{Q\cap A\mid Q\in \Sigma\}$.
This is the initial $\sigma$-algebra for the inclusion $\iota:A\to S$ since $\iota^{-1}[Q]=Q\cap A\in \Sigma\rest A$, and therefore $\iota:(A,\Sigma\rest A) \to (S,\Sigma)$ is measurable. By applying the functor $\Delta$, we obtain the measurable and injective map $\Delta \iota:(\Delta(A),\Delta(\Sigma\rest A))\to (\Delta(S),\Delta(\Sigma))$ given by $(\Delta \iota)(\tilde{\mu})=\tilde{\mu}\circ \iota^{-1}$. We note that for $Q \in \Sigma$,
\begin{align}\label{eq:delta-i-preimage-base-set}
	\begin{split}
		(\Delta \iota)^{-1}[\Delta^{<q}(Q)]&=\{\tilde{\mu}\in \Delta(A)\mid 
		\Delta \iota(\tilde{\mu}) \in \Delta^{<q}(Q)\} \\
		&=\{\tilde{\mu}\in \Delta(A)\mid \tilde{\mu}(\iota^{-1}(Q))<q\}\\
		&=\{\tilde{\mu}\in \Delta(A)\mid \tilde{\mu}(Q\cap A)<q\} \\
		&=\Delta^{<q}(Q\cap A)\in \Delta(\Sigma\rest A).
	\end{split}
\end{align}

\begin{lemma}\label{lem:delta-i-preimage}
  If $A\sbq S$ and $\Gamma \sbq \Sigma$, then $(\Delta
  \iota)^{-1}[\Delta(\Gamma)]=\Delta(\Gamma\rest A)$. In particular,
  $\Delta(\Sigma\rest A) =(\Delta \iota)^{-1}[\Delta(\Sigma)]$.
\end{lemma}
\begin{proof}
  \begin{align*}
    (\Delta \iota)^{-1}[\Delta(\Gamma)]&= (\Delta
    \iota)^{-1}[\sigma(\{\Delta^{<q}(Q)\mid Q\in \Gamma,\ q\in \Q\cap[0,1]\})] &&
    \text{Definition of $\Delta(\Gamma)$}\\
    &=\sigma((\Delta \iota)^{-1}[\{\Delta^{<q}(Q) \mid Q\in\Gamma,\ q\in \Q\cap[0,1]\}]) \\
    &= \sigma(\{\Delta^{<q}(Q\cap A) \mid Q\in\Gamma,\ q\in \Q\cap[0,1] \}) && 
    \text{Equation~\eqref{eq:delta-i-preimage-base-set}}\\
    &= \sigma(\{\Delta^{<q}(Q')\mid Q'\in \Gamma\rest A,\ q\in \Q\cap[0,1] \}) &&
    \text{Definition of $\Gamma\rest A$} \\
    &= \Delta(\Gamma\rest A). && \hfill \qedhere
  \end{align*}
\end{proof} 

\subsection{Thick subspaces}
\label{sec:thick-subspaces}

We recall the following concept from \cite{Halmos}: A subset $A$ of a measure space $(X,\Sigma,\mu)$ is said to be \emph{thick} 
with respect to $\mu$ if its outer measure, denoted by $\mu^*(A)$, equals the total measure $\mu(X)$. 
This is equivalent to the condition $\forall F\in \Sigma \; (F\cap A=\emptyset \implies \mu(F)=0)$. 
For example, any Lebesgue-measurable subset of the real line having full measure (i.e., having a null measure complement) is thick. 
There are, although, thick nonmeasurable sets (even Vitali sets: those formed by representatives for the equivalence relation $x-y\in\Q$).

\begin{theorem}[{\cite[p.75]{Halmos}}]\label{thm:A-thick-measure-space}
  Let $A$ be a thick subset of a measure space $(S,\Sigma,\mu)$. Then the
  stipulation $\mu_A(E\cap A)\defi\mu(E)$ for $E\in\Sigma$ well-defines a
  measure space $(A,\Sigma\rest A,\mu_A)$.
\end{theorem}

If $A$ is measurable, the condition of being thick with respect to $\mu$ is
equivalent to $\mu(A)=\mu(S)$. Furthermore, for $E \in \Sigma$ we have
\[
  \mu_A(E\cap A) = \mu(E) = \mu(E\cap A)+\mu(E\cap A^c) = \mu(E\cap A),
\]
so $\mu_A = \mu\rest_{\Sigma\rest A}$.

\begin{lemma}
	If $A$ is measurable, then $\Delta \iota$ is a bijection between $\Delta(A,\Sigma\rest A)$ and $\{\mu\in \Delta(S) \mid \mu(A) = \mu(S)\}$ with inverse $j$ defined by $j(\mu) = \mu_A$.\qed
\end{lemma}

We observe that $\dom(j)=\{\mu \in \Delta(S)\mid \mu(A) = \mu(S)\}$ belongs to $\Delta(\Sigma)$ because it is the set where the two measurable functions $\mathrm{ev}_A$ and $\mathrm{ev}_S$ coincide.
If $E\in \Sigma$, 
\begin{align*}
  j^{-1}[\Delta^{<q}(E\cap A)]& = \{\mu \in \dom(j) \mid 
    \mu_A\in \Delta^{<q}(E\cap A)\} \notag \\
  & = \{\mu \in \dom(j) 
    \mid \mu_A(E\cap A)<q \} \label{eq:j-measurable} \\
  & = \{\mu \in \dom(j) \mid \mu(E)<q\} = 
    \Delta^{<q}(E)\cap \mathrm{dom}(j) \in \Delta(\Sigma). \notag
\end{align*}

Suppose that $A\sbq S$ is thick for all $\mu \in \bigcup\{\Tfont{T}_a(s)\mid s\in A, \; a\in L\}$. By Theorem~\ref{thm:A-thick-measure-space}, for each $s\in A$ we have a set of measures 
\begin{equation*}
	(\Tfont{T}_a\rest A)(s)\defi\{\mu_A\in \Delta(A)\mid \mu \in 
	\Tfont{T}_a(s)\}\sbq \Delta(A)
\end{equation*}
and we can consider the tuple 
\begin{equation}\label{eq:sub-NLMP}
	\lmp{A}\defi(A,\Sigma\rest A,\{\Tfont{T}_a\rest A \mid a\in L\}).
\end{equation}

From now on, we will focus on \emph{measurable} thick subsets.
\begin{prop}\label{prop:substruct-is-NLMP}
	Let $A \subseteq S$ be measurable such that for all $s \in A$, $a \in L$, and any measure $\mu \in \Tfont{T}_a(s)$, $A$ is a thick subset of $(S, \Sigma, \mu)$. 
	Then $\lmp{A}$ in \eqref{eq:sub-NLMP} is an NLMP.
\end{prop}
\begin{proof}
	If $\mu \in \Tfont{T}_a(s)$, then $\mu = \mu_A \circ \iota^{-1}$, and since $\Tfont{T}_a(s) \in \Delta(\Sigma)$, we have 
	\begin{align*}
		(\Tfont{T}_a \rest A)(s) &= \{\mu_A \in \Delta(A) \mid \mu \in \Tfont{T}_a(s)\} = \{\mu_A \in \Delta(A) \mid \Delta \iota(\mu_A) \in \Tfont{T}_a(s)\} \\
		&= \{\mu \in \Delta(A) \mid \Delta \iota(\mu) \in \Tfont{T}_a(s)\} = (\Delta \iota)^{-1}[\Tfont{T}_a(s)] \in \Delta(\Sigma \rest A).
	\end{align*}
	
	Furthermore, $\Tfont{T}_a \rest A: A \to \Delta(\Sigma \rest A)$ is a measurable function. Indeed, let $\tilde{D} \in \Delta(\Sigma \rest A)$ and $D \in \Delta(\Sigma)$ such that $\tilde{D} = (\Delta \iota)^{-1}[D]$, then 
	\begin{align*}
		(\Tfont{T}_a \rest A)^{-1}[H_{\tilde{D}}] &= \{s \in A \mid (\Tfont{T}_a \rest A)(s) \cap \tilde{D} \neq \emptyset\} \\
		&= \{s\in A\mid \exists \mu\in \Tfont{T}_a(s), \; \Delta \iota(\mu_A)\in D\} \\
		&= \{s \in A \mid \exists \mu\in \Tfont{T}_a(s), \; \mu\in D\} \\ 
		&= \{s \in A \mid \Tfont{T}_a(s)\cap D\neq \emptyset\} \\
		&= \Tfont{T}_a^{-1}[H_D] \cap A \in \Sigma \rest A. \qedhere
	\end{align*}
\end{proof}

\begin{definition}\label{def:sub-NLMP}
  Let $\lmp{S}$ be an NLMP and let $A \subseteq S$ be measurable such that for
  all $s \in A$, $a \in L$, and any measure $\mu \in \Tfont{T}_a(s)$, $A$ is a
  thick subset of $(S, \Sigma, \mu)$.  We say that $\lmp{A}$ defined in
  \eqref{eq:sub-NLMP} is a \emph{substructure} (or sub-NLMP)
  of $\lmp{S}$ and denote it as $\lmp{A}\leqslant \lmp{S}$.
\end{definition}

\subsection{Upward coherence of bisimulation}
\label{sec:upward-coherence-bisimulation}
\begin{prop}\label{prop:S,s-bisim-A,s}
	Let $\lmp{S}$ be an NLMP and $\lmp{A}$ a substructure. 
	If $s \in A$, then $(\lmp{S}, s) \sim_\s^\times 
	(\lmp{A}, s)$.
\end{prop}
\begin{proof}
	It suffices to show that the relation $R \subseteq S \times A$ given by 
	$R \defi \id \rest A$ is an external bisimulation between $(\lmp{S}, s)$ and $(\lmp{A}, s)$. 
	Let $(x, x) \in R$ and $(E, E')$ be an $R$-closed pair such that $E \in \Sigma$ and $E' \in \Sigma \rest A$. 
	Then $E \cap A = E'$. 
	Now, if $\mu \in \Tfont{T}_a(x)$, since $A$ is thick with respect to $\mu$, we have $\mu(E) = \mu(E \cap A) = 
	\mu_A(E')$. 
	Conversely, if $\mu_A \in \Tfont{T}_a \rest A(x)$, then $\mu_A(E') = \mu_A(E \cap A) = \mu(E)$.
\end{proof}

We now turn to the notions of bisimulations between substructures.
We first need to say something about the 
measurable closed sets for a relation on the subspaces.

\begin{lemma}\label{lem:closed-sets-and-pairs-S-S'}
		Let $A \subseteq S$, $A' \subseteq S'$ and $R \subseteq A \times A'$. If $(Q,Q')$ is a $(\Sigma, \Sigma')$-measurable $R$-closed pair, then $(Q\cap A,Q'\cap A')$ is a $(\Sigma\rest A, \Sigma'\rest A')$-measurable $R$-closed pair.
		Conversely, if $(E,E')$ is a $(\Sigma\rest A, \Sigma'\rest A')$-measurable $R$-closed pair, then $E=Q\cap A$ and $E'=Q'\cap A'$ for some $(\Sigma, \Sigma')$-measurable $R$-closed pair $(Q,Q')$.\qed
\end{lemma}

\begin{lemma}\label{lem:ext-in-A-iff-ext-in-S}
 Let $\lmp{A}\leqslant \lmp{S}$, $\lmp{A'}\leqslant\lmp{S'}$ be two substructures and $R\sbq A\times A'$.
 Then, $R$ is an external state bisimulation between $\lmp{A}$ and $\lmp{A}'$ if and only if it is an external state bisimulation between $\lmp{S}$ and $\lmp{S}'$.
\end{lemma}
\begin{proof}
  We only show the $(\Rightarrow)$ direction, since the converse is very similar.	Let
  $s\rR t$ (necessarily we have $s\in A$ and $t\in A'$), assume $R$ is a
  bisimulation between $\lmp{A}$ and $\lmp{A}'$, and let $\mu\in
  \Tfont{T}_a(s)$.
	Then $\mu_A\in (\Tfont{T}_a\rest A)(s)$.  
	Since $s\mathrel{R} t$, there exists 
	$\mu'\in \Tfont{T}'_a(t)$ such that 
	$\mu_A \mathrel{\overline{R}} \mu'_{A'}$.  
	We verify that $\mu \mathrel{\overline{R}} \mu'$.
	Using (the first part of) Lemma~\ref{lem:closed-sets-and-pairs-S-S'}, if $(Q,Q')$ is a $(\Sigma, \Sigma')$-measurable $R$-closed pair, then $(Q\cap A,Q'\cap A')$ is a $(\Sigma\rest A, \Sigma'\rest A')$-measurable $R$-closed pair.
	Hence,
	\[
	\mu(Q)=\mu_A(Q\cap A)=\mu'_{A'}(Q'\cap A')=\mu'(Q').
	\]
	Repeating this argument for the zag condition, we conclude that $R$ is an external state bisimulation between $\lmp{S}$ and $\lmp{S}'$.  
\end{proof}

In the case where $\lmp{S}=\lmp{S}'$, we can directly compare external and internal bisimulations because they are of the same type.
Some care is required, however, since the internal notion requires symmetry.
In this setting, Lemma~\ref{lem:S-equals-S'-case} ensures that the two notions coincide on reflexive and symmetric relations.

\begin{corollary}\label{cor:bisim-in-NLMP-A} 
	Let $\lmp{A} \leqslant \lmp{S}$ and $R \subseteq A \times A$.
	\begin{enumerate}
		\item\label{item:sigma-R-rest} $\Sigma(R) \rest A = \Sigma \rest A(R)$.
		\item\label{item:state-on-A-iff-on-S} If $R$ is symmetric and reflexive, then $R$ is a state bisimulation on $\lmp{A}$ if and only if it is a state bisimulation on $\lmp{S}$.
	\end{enumerate}
\end{corollary}
\begin{proof}
	The first item follows from Lemmas~\ref{lem:closed-vs-closed-pair}(\ref{item:E-closed-iff-EE-closed}) and \ref{lem:closed-sets-and-pairs-S-S'}.
	The second one is a consequence of Lemmas~\ref{lem:S-equals-S'-case} and \ref{lem:ext-in-A-iff-ext-in-S}.
\end{proof}

The next result relates external state
bisimulations between substructures of the same NLMP with its
internal bisimulations.

\begin{lemma}\label{lem:ext-bisim-NLMP-internal}
	Let $\lmp{A}$ and $\lmp{A}'$ be two substructures of $\lmp{S}$.
	If $R \subseteq A \times A'$ is an external state bisimulation, then
	$R \cup R^{-1}$ is a state bisimulation on $\lmp{S}$.
\end{lemma}
\begin{proof}
  By Lemma~\ref{lem:ext-in-A-iff-ext-in-S}, $R$ and $R^{-1}$ are external
  state bisimulations between $\lmp{S}$ and $\lmp{S}$. Their union is also a
  \textit{symmetric} external bisimulation. Hence it is also internal by the
  comment after Lemma~\ref{lem:closed-vs-closed-pair}.
\end{proof}

\subsection{Restricting bisimulations}
\label{sec:restricting-bisimulations}

So far, we have worked with relations defined on the base spaces of the substructures, that is, subsets of $A\times A'$.
We now turn to the study of restrictions of relations on $S$ to
substructures.
If $R\sbq S\times S$ and $A, A'\subseteq S$ are fixed, we let $R{\downarrow} \defi R \cap A\times A'$.

\begin{example}\label{exm:counterexm-NLMP-state-restriction}
	We construct a process $\lmp{S}$ showing that the restriction of a
	state bisimulation on $\lmp{S}$ to a substructure $\lmp{A}$ does not
	necessarily yield a state bisimulation on $\lmp{A}$.
	
	Consider the interval $I = (0,1)$ equipped with the Lebesgue measure
	$\leb$, and let $V \subseteq (0,\tfrac{1}{2})$ be a $\leb$-measurable
	non-Borel subset.
	Let $Q \subseteq V$ and $Q' \subseteq I \setminus V$ be Borel sets such
	that $\leb(Q) = \leb(V)$ and $\leb(Q') = \leb(I \setminus V)$, and fix
	$c \in (0,1)$.
	On the measurable space $(S,\Sigma) \defi (I \cup \{s,t\},
	\sigma(\Borel(I) \cup \Power(\{s,t\})))$ we define the transitions
	\begin{align*}
		\tau(s) &= \leb, \\
		\tau(t) &= c\,\leb\rest(0,\tfrac{1}{2}) +
		(1-c)\,\leb\rest[\tfrac{1}{2},1).
	\end{align*}
	If $R \defi \Rel(\{V,\{s,t\}\})$, then
	$\Sigma(R) = \{S,\emptyset,I,\{s,t\}\}$ and therefore $R$ is a state
	bisimulation.
	
	Define $A \defi Q \cup Q' \cup \{s,t\}$. Then $A$ is measurable and	thick for all the transitions.
	The relation $R{\downarrow}$ has equivalence classes
	$Q$, $Q'$, and $\{s,t\}$.
	Moreover, the following strict inclusion holds:
	\[
	\Sigma(R)\rest A
	= \{A,\{s,t\},Q \cup Q',\emptyset\}
	\subsetneq
	\sigma(\{Q,Q',\{s,t\}\})
	= \Sigma\rest A(R{\downarrow}).
	\]
	It follows that $R{\downarrow}$ is not a state bisimulation on $\lmp{A}$,
	since $(s,t) \in R{\downarrow}$, but $Q \in \Sigma\rest A(R{\downarrow})$ satisfies
	\[
	\tau(s)(Q) = \leb(V) \neq c\,\leb(V) = \tau(t)(Q).
	\]
\end{example}

A problem in the previous example is that the bisimulation $R$ is not
``well behaved.''

\begin{lemma}[{\cite[Lemma~5.4.6]{Gao2008InvariantDS}}]
	\label{lem:invariant-analytic-separation}
	Let $E$ be an analytic equivalence relation on a standard Borel space
	$X$.
	Let $B,C \subseteq X$ be two disjoint $E$-invariant (i.e., $E$-\emph{closed}) analytic sets.
	Then there exists a $E$-invariant Borel set $D$ such that
	$B \subseteq D$ and $D \cap C = \emptyset$.
\end{lemma}

\begin{lemma}[{\cite[Cor.~2, p.~73]{arveson1976invitation}}]
	\label{lem:seq-borel-functions-quotient}
	Let $X$ be an analytic measurable space and let $R$ be an equivalence
	relation on $X$.
	Suppose there exists a sequence of Borel functions
	$f_1,f_2,\dots : X \to \mathbb{R}$ such that for all $x,y \in X$,
	\[
	x \mathrel{R} y \iff \forall n\; f_n(x) = f_n(y).
	\]
	Then the quotient space $X/R$ is analytic.
\end{lemma}

In particular, the hypothesis about the relation is satisfied for a smooth $R$ given as $\Rel(\Fam)$ for a countable family of Borel sets. Simply take the characteristic functions of the sets in $\Fam$.

\begin{lemma}\label{lem:restriction-smooth-is-NLMP-state-bisim}
	Suppose that $\lmp{A}$ is a substructure of $\lmp{S}$ and that
	$R \subseteq S \times S$ is a smooth state bisimulation.
	Then $R{\downarrow}$ is a state bisimulation on
	$\lmp{A}$.
\end{lemma}

\begin{proof}
	By Lemma~\ref{lem:seq-borel-functions-quotient}, the quotient space $S/R$
	is analytic, and by construction, the canonical projection $\pi : S \to S/R$ is measurable.
	Let $(s,s') \in R{\downarrow}$ and
	$\mu_A \in (\Tfont{T}_a\rest A)(s)$.
	Let $\mu' \in \Tfont{T}_a(s')$ be such that $\mu \mathrel{\overline{R}} \mu'$.
	We show that $\mu_A \mathrel{\overline{R{\downarrow}}} \mu'_A$.
	By Corollary~\ref{cor:bisim-in-NLMP-A}(\ref{item:sigma-R-rest}),
	$(\Sigma\rest A)(R{\downarrow}) = \Sigma(R{\downarrow})\rest A$, and since $A$ is
	measurable, $(\Sigma\rest A)(R{\downarrow}) \subseteq \Sigma(R{\downarrow})$.
	If $W \in (\Sigma\rest A)(R{\downarrow})$, then $W \in \Sigma$ and $W \subseteq A$.
	Define
	$B \defi \pi^{-1}[\pi[W]]$ and $C \defi \pi^{-1}[\pi[A\setminus W]]$.
	Both sets are $R$-invariant analytic and satisfy $B \cap C = \emptyset$
        (by transitivity of $R$ and the fact that $W$ is $R{\downarrow}$-closed).
	By Lemma~\ref{lem:invariant-analytic-separation}, there exists
	$D \in \Sigma(R)$ containing $B$ and disjoint from $C$.
	Hence $W = D \cap A \in \Sigma(R)\rest A$, and therefore
	$\mu_A(W) = \mu(D) = \mu'(D) = \mu'_A(W)$.	
\end{proof}

The reliance of this lemma on the properties of standard Borel spaces is not merely a convenience. 
Indeed, outside this class of spaces, the restriction property can fail even for simple relations. In \cite[Exm.~3.37]{survey_bisim}, an LMP is constructed using non-measurable sets in which two states $s,t$ are state bisimilar in the total process, yet they are not related by any external bisimulation between two substructures containing them. 
This pathology demonstrates that without the structural guarantees of standard Borel spaces, global bisimulations do not necessarily induce local ones. 

The following concept is related to “partitions of the unity”. It encodes a
countable family of measurable sets that decompose a set $A$ but only from the
point of view of measure: The “cells” of the partition need not be disjoint;
their intersections only need to be null. The same goes for their union (more
precisely, for its complement).
\begin{definition}\label{def:partition-in-measure}
  Let $A\in \Sigma$. We say that a family $\{B_i \mid i \in I \}$ of measurable
  subsets of $A$ is a
  \emph{partition in measure} of $A$ if
  \begin{itemize}
  \item $I$ is countable;
  \item $\nu(A) = \nu(\union \{B_i \mid i \in I \})$;
  \item $i\neq j \implies \nu(B_i\cap B_j) = 0$.
  \end{itemize}
\end{definition}
Note that if $U\subseteq A$ is measurable, $\{U\cap B_i \mid i \in I \}$ is also a
partition in measure of $U$.

\begin{lemma}\label{lem:partition-sum}
  If $\{B_i \mid i \in I \}$ is a partition in measure of $A$, then
  \[
    \nu(A) = \nu_A\bigl(\union\{B_i \mid i \in I \}\bigr) 
    = \sum \{\nu_A(B_i) \mid i \in I \}. \qed
  \]
\end{lemma}

Most of our results will depend on relations $R$ being \textit{measurably generated}, that is, of the form $R = \Rel(\Lambda)$ for some $\Lambda \subseteq \Sigma$. 
All such relations satisfy $R=\Rel(\Sigma(R))$.
 
\begin{lemma}\label{lem:R-bar-interpolation}
  Let $R$ be an equivalence relation on $S$ such that $R=\Rel(\Sigma(R))$,
  $A,A'\subseteq S$, $\nu, \nu'\in \Delta(S)$ such that $A$, $A'$ are thick for
  $\nu$, $\nu'$ respectively. Assume moreover that there exist $B_i\in
  \Sigma\rest A$  and
  $B_i'\in \Sigma\rest A'$ (for $i\in I$) such that
  \begin{enumerate}
  \item\label{item:1} $\{B_i \mid i \in I \}$ is a partition in measure of $A$;
  \item\label{item:2} $\{B_i' \mid i \in I \}$ is a partition in measure of $A'$;
  \item\label{item:3} for all $i\in I$, $\nu_A(B_i) = \nu'_{A'}(B_i')$; and
  \item\label{item:4}
    for each $i\in I$ there is an $R$-class $C$ %
    such that
    \begin{enumerate}
    \item\label{item:4a} $C \cap A \subseteq B_i$ and for every $Q\in\Sigma$, if
      $C\cap A \subseteq Q $, then $\nu_A(Q\cap B_i) = \nu_A(B_i)$; and
    \item\label{item:4b} analogous to Item~\ref{item:4a} but with $B_i'$ and $A'$.
    \end{enumerate}
  \end{enumerate}
  Then $\nu_A \mathrel{\overline{R{\downarrow}}} \nu'_{A'}$.
\end{lemma}
\begin{proof}
  In the Appendix (Lemma~\ref{lem:R-bar-interpolation-appendix}).
\end{proof}

We say that a measure  $\nu \in \Delta(S)$ is \emph{countably} (resp., \emph{finitely})
\emph{supported} or that it has \emph{countable} (\emph{finite}) \emph{support} if there is a
countable (finite) subset $N\subseteq S$ such that $\nu(S\sm N) = 0$. For such a
measure $\nu$, its \emph{support} is
\[
  \supp\nu \defi \{ s \in S \mid \nu(\{s\}) > 0 \}.
\]
\begin{lemma}\label{lem:R-bar-interpolation-countable-supp}
  Let $R$ be an equivalence relation on $S$ such that $R=\Rel(\Sigma(R))$,
  $A,A'\in\Sigma$, $\nu, \nu'\in \Delta(S)$ having countable supports and such
  that $A$, $A'$ are thick for $\nu$, $\nu'$ respectively. If 
  $\nu \mathrel{\overline{R}} \nu'$, then 
  $\nu_A \mathrel{\overline{R{\downarrow}}} \nu'_{A'}$.
\end{lemma}
\begin{proof}
  For every $x,y \in S$ such that $x\not\rR y$, let $Q_{x,y}\in \Sigma(R)$ such that
  $Q_{x,y}\cap \{ x, y \} = \{ x \}$. Let $C$ be an $R$-equivalence class and any
  $x\in C$, define
  \[
    Q_C \defi \bigcap\{ Q_{x,y} \mid y \in (\supp\nu \cup \supp\nu') \sm C \}.
  \]
  We have that $C\subseteq Q_C\in\Sigma(R)$ and
  \begin{equation}
    \label{eq:QC_inter_sup_eq_C_inter_sup}
     Q_C \cap (\supp\nu \cup \supp\nu') = C \cap (\supp\nu \cup
     \supp\nu').
  \end{equation}
  We claim that 
  \[
    \{Q_C \cap A \mid C \text{ eq.~class}, C\cap (\supp\nu \cup
     \supp\nu') \neq \emptyset \}
  \]
  and analogously with $A'$, are partitions in
  measure satisfying the hypothesis of Lemma~\ref{lem:R-bar-interpolation}, and
  hence we obtain the conclusion.

  Both families are countable and consist of measurable sets. Moreover, if $C$
  is an $R$-class, we have
  \begin{align*}
    \nu_A(Q_C\cap A)
    &= \nu(Q_C) && A \text{ thick,}\\
    &= \nu'(Q_C) && Q_C \text{ is $R$-closed,}\\
    &= \nu'_{A'}(Q_C\cap A') && A' \text{ thick.}
  \end{align*}

  Finally, take $D\neq C$ a different $R$-class. Hence
  \begin{align*}
    \nu((Q_C\cap A) \cap (Q_D\cap A))
    &= \nu((Q_C\cap \supp\nu \cap A) \cap (Q_D\cap \supp\nu \cap A)) \\
    &\leq \nu((Q_C\cap (\supp\nu \cup \supp\nu') \cap A) \cap \dots) \\
    &= \nu((C\cap (\supp\nu \cup \supp\nu') \cap A) \cap (D\cap \dots)) \\
    &= 0,
  \end{align*}
  since $C\cap D = \emptyset$. We conclude that, if $\{ C_i \mid i\in I \}$ is an
  enumeration of the equivalence classes $C$ such that $C\cap (\supp\nu \cup
  \supp\nu') \neq \emptyset$, $B_i \defi Q_{C_i}\cap A$, and $B_i' \defi  Q_{C_i}\cap A'$,
  \[
    \{ B_i \mid i \in I \} \text{ and }     \{ B_i' \mid i \in I \}
  \]
  are partitions in measure of $A$ and $A'$, respectively, satisfying
  Item~\ref{item:3} from the hypothesis of Lemma~\ref{lem:R-bar-interpolation}.

  Now, to show Item~\ref{item:4} we let $C\defi C_i$ for each $i\in I$, we get
  \[
    C \cap A \subseteq Q_C\cap A = B_i,
  \]
  and if $C\cap A \subseteq Q \in \Sigma$, the reader can verify that $\nu_A(Q\cap B_i) = \nu_A(B_i)$, which concludes the proof.
\end{proof}

We say that a relation $R\sbq X \times X'$ is $z$-closed if, for all $x,y \in X$ and $x',y'\in X'$, the conditions $x \rR x'$, $y \rR x'$ and $y \rR y'$ imply $x \rR y'$.
This concept is also known as difunctionality or rectangularity, and can be viewed as a generalization of transitivity to relations between different spaces.

\begin{lemma}\label{lem:R_descends_to_bisim_between_substruc}
  Assume that for all $a\in L$, $s\in S$, and $\mu \in \Tfont{T}_a(s)$, $\supp\mu$ is countable. Let $R$ be a state bisimulation on $\lmp{S}$
  such that $R=\Rel(\Sigma(R))$.

  Hence for all $s,s'\in S$ such that $s \rR s'$ and for all 
  substructures $\lmp{A}$ and $\lmp{A}'$ such that $s\in A$ and $s' \in A'$,
  $R{\downarrow}$ is a $z$-closed external state bisimulation (which relates $s$ to $s'$) between $\lmp{A}$ and $\lmp{A}'$.
\end{lemma}
\begin{proof}
  Assume $R$ as in the hypotheses. Since $R$ is an equivalence relation,
  $R{\downarrow}$ is immediately $z$-closed.
  To see that $R{\downarrow}$ is an external state bisimulation, assume
  $t \mathrel{R{\downarrow}} t'$ and $\tilde\nu \in (\Tfont{T}_a\rest
  A)(t)$. Then $\tilde{\nu}=\nu_A$ for some $\nu\in \Tfont{T}_a(t)$.
  Since $R$ is a state bisimulation, there
  exists $\nu'\in \Tfont{T}_a(t')$ such that $\nu \mathrel{\overline{R}}
  \nu'$. By Lemma~\ref{lem:R-bar-interpolation-countable-supp}, we have $\nu_A
  \mathrel{\overline{R{\downarrow}}} \nu'_{A'} \in(\Tfont{T}_a\rest A')(t')$,
  and we are done.
\end{proof}

\section{Pointmass NLMP}\label{sec:pointmass-nlmp}

\subsection{Discrete processes}
\label{sec:discrete-processes}

For the rest of the paper we will work with processes $\lmp{S}$ where, for every $s\in S$ and $a\in L$, all the transitions $\mu\in \Tfont{T}_a(s)$ are discrete, that is, can be written as a sub-convex combination $\sum_k r_k\delta_{x_k}$ for some $\{r_k\}_k\sbq[0,1]$ and $\{x_k\}_k\sbq S$. 
For the subfamily of finitely supported processes, which we informally call
“pointmass processes”, we will prove that bisimilarity is analytic
(Theorem~\ref{th:state-bisim-finit-supp}) by producing a Borel definition of bisimulation.

\begin{definition}
  Let $\lmp{S}=(S,\Sigma,\{\Tfont{T}_a\mid a\in L\})$ be an image-countable NLMP over a
    standard Borel
    space. We say that $\lmp{S}$ is \emph{discrete} if every $\mu \in \Tfont{T}_a(\cdot)$ has countable support.
\end{definition}	

Every countable NLMP (endowed with the powerset $\sigma$-algebra) is automatically discrete.

\begin{lemma}\label{lem:bisim-equiv-between-generated}
  Let $\lmp{S}=(S,\Sigma,\{\Tfont{T}_a\mid a\in L\})$ be a discrete NLMP, $s,s'\in S$ and $\lmp{A}$, $\lmp{A}'$ two substructures such that $s\in A$ and $s' \in A'$.	
    Then $s\sim_\s s'$ if and only if there exists a $z$-closed external state bisimulation $R\sbq A\times A'$ between $\lmp{A}$ and $\lmp{A}'$ such that $s\mathrel{R} s'$.
\end{lemma}
\begin{proof}
  The $(\Rightarrow)$ direction follows from Lemma~\ref{lem:R_descends_to_bisim_between_substruc} applied to ${\sim_\s}$.
  For the converse, apply Lemma~\ref{lem:ext-bisim-NLMP-internal}.
\end{proof}

From Lemma~\ref{lem:S-equals-S'-case}, we know that within a single (general) NLMP,
\[
s \sim_{\s} t
\iff
(\lmp{S},s) \sim_{\s}^{\times} (\lmp{S},t).
\]
Moreover, Lemma~\ref{lem:bisim-equiv-between-generated} shows that,
for discrete NLMP, this equivalence extends to arbitrary
substructures: for any substructures $\lmp{A}$ and $\lmp{A}'$ of
$\lmp{S}$ containing $s$ and $s'$ respectively,
\begin{equation}\label{eq:bisim_in_S_iff_in_substructures}
	s \sim_{\s} t \iff (\lmp{A},s) \sim_{\s}^{\times} (\lmp{A}',t).
\end{equation}

With the following definition, we aim to represent, for an NLMP $\lmp{S}$ and a state $s\in S$, the set $A_s$ consisting of $s$ and all states accessible from $s$ (successors of $s$, successors of successors of $s$, and so forth, within a finite number of steps).

\begin{definition}\label{def:general_A_s}
	Let $\lmp{S}=(S,\Sigma,\{\Tfont{T}_a\mid a\in 
	L\})$ be a discrete NLMP. For $s \in S$, we define 
	\begin{enumerate}
		\item $\tilde{\Tfont{T}}_a(s)\defi \bigcup\{\supp \mu\mid \mu \in \Tfont{T}_a(s)\}$.
		\item $A_s\defi \bigcup_{n\in \omega}A_n(s)$, where the family 
		$\{A_n(s)\}_{n\in \omega}$ is given recursively by: 
		\begin{itemize}
			\item $A_0(s)=\{s\}$, 
			\item $A_{n+1}(s)=A_n(s)\cup \bigl(\bigcup \{\tilde{\Tfont{T}}_a(x) \mid x\in A_n(s), \, a\in L\}\bigr)$.
		\end{itemize}
	\end{enumerate}
\end{definition}

Notice that each $A_s\subseteq S$ is measurable, since it is a countable set.
Also, for all $x \in A_s$, $a \in L$, and $\mu \in \Tfont{T}_a(x)$, we have $\mu(S)=\mu(\supp(\mu))\leq\mu(A_s)\leq\mu(S)$. 
So $A_s$ is a thick subset of $(S, \Sigma, \mu)$.
It follows that $\lmp{A}_s$ is a substructure as in Definition~\ref{def:sub-NLMP}.

\begin{remark}\label{rem:about_As}
  \begin{enumerate}
  \item
    $\lmp{A}_s$ is the smallest substructure of $\lmp{S}$ containing $s$.
  \item
    The substructures $\lmp{A}_s$ show that it is possible to restrict a measurably generated bisimulation on a discrete NLMP to \emph{countable} substructures $\lmp{A}$ and $\lmp{A}'$ to obtain an external state bisimulation (see  Lemma~\ref{lem:R_descends_to_bisim_between_substruc}).
  \item\label{item:T_empty}
    $\tilde{\Tfont{T}}_a(s) = \emptyset \iff \Tfont{T}_a(s) = \emptyset \vee \Tfont{T}_a(s) = \{0\}$.
  \end{enumerate}
\end{remark}

\begin{lemma}\label{lem:A_s_saturated}
	Let $\lmp{S}$, $\lmp{S}'$ be two discrete NLMPs.
	Let $R\sbq S\times S'$ be a $z$-closed external state bisimulation such that $s\mathrel{R} s'$.
	Then, for all $n \in \omega$, $A_n(s') \sbq R[A_n(s)]$ and $A_n(s) \sbq R^{-1}[A_n(s')]$.
\end{lemma}
\begin{proof}
	We only prove that $A_n(s')\sbq R[A_n(s)]$ for all $n$. 
	The case $n=0$, $A_0(s')=\{s'\}\sbq R[A_0(s)]$, follows by hypothesis.
	
	Suppose this is true for $n$, we want to show that is also true for $A_{n+1}(s')=A_n(s')\cup \bigl(\bigcup \{\tilde{\Tfont{T}}'_a(t') \mid t'\in A_n(s'), \, a\in L\}\bigr)$.
	Let $x' \in \tilde{\Tfont{T}}'_a(t')$ for some $t'\in A_n(s')$.
	By IH, $t'\in R[A_n(s)]$, so there is $t\in A_n(s)$ such that $t\mathrel{R}t'$.
	
	Suppose $x'\in \supp \mu'$ for $\mu' \in \Tfont{T}'_a(t')$. 
	Since $t\mathrel{R} t'$, there exists $\mu \in \Tfont{T}_a(t)$ such that $\mu \mathrel{\bar{R}} \mu'$.
	We want to check that $R^{-1}[x']\cap \supp \mu\neq \emptyset$.
	Suppose, for contradiction, that it is empty and consider the pair $(R^{-1}[x'],\{x'\}\cup R[R^{-1}[x']])$.
	This is a measurable $R$-closed pair such that $\mu'(\{x'\}\cup R[R^{-1}[x']])\geq \mu'(\{x'\})> 0 = \mu(R^{-1}[x'])$, but this contradicts $\mu \mathrel{\bar{R}} \mu'$.
	Thus, we know that there exists $x\in \supp \mu \sbq \tilde{\Tfont{T}}_a(t)$ such that $x\mathrel{R} x'$.
	As $\tilde{\Tfont{T}}_a(t)\sbq A_{n+1}(s)$, we conclude $x'\in R[A_{n+1}(s)]$.
\end{proof}

\begin{remark}\label{rem:external_is_traditional_countable_NLMP}
	Note that by the proof of case $n=1$ of the previous lemma, an external bisimulation behaves like a standard bisimulation in the following sense: 
	If $s\mathrel{R} s'$, then for every $t\in \tilde{\Tfont{T}}_a(s)$, there exists $t'\in \tilde{\Tfont{T}}'_a(s')$ such that $t\mathrel{R} t'$ (and the analogous ``zag'' statement).
\end{remark}

One immediate consequence of Lemma~\ref{lem:A_s_saturated} is that $A_{s'} \sbq R[A_s]$ and $A_s \sbq R^{-1}[A_{s'}]$. Also notice that, if the NLMPs considered are exactly the substructures $\lmp{A}_s$ and $\lmp{A}_{s'}$, then we get an equality, and we say that the pair $(A_s, A_{s'})$ is $R$-saturated.

The next results will allow us to simplify the universal quantification in the
definition of $\mathrel{\bar{R}}$ to a finite domain for finitely supported processes.
\begin{lemma}\label{lem:barR_characterization_countable_supp}
	Let $\lmp{A}$, $\lmp{A}'$ be two countable NLMPs over separable spaces.
	Let $R\sbq A\times A'$ be a $z$-closed external state bisimulation such that $s\mathrel{R} s'$.
	Then, for every $\mu\in \Tfont{T}_a(s)$ and $\mu'\in \Tfont{T}'_a(s')$ it holds:
		\begin{align*}
			\mu \mathrel{\bar{R}} \mu' \iff & \forall x \in \supp \mu,\,
			\mu(R^{-1}[R[x]] \cap \supp \mu) = \mu'(R[x] \cap \supp \mu') \land {} \\
			& \forall y' \in \supp \mu',\,
			\mu(R^{-1}[y'] \cap \supp \mu) = \mu'(R[R^{-1}[y']] \cap \supp \mu').
		\end{align*}
\end{lemma}
\begin{proof}
	$(\Rightarrow)$ It is enough to see that for every $x\in \supp \mu$ and $y'\in \supp \mu'$, the pairs $(R^{-1}[R[x]], R[x])$ and $(R^{-1}[y'], R[R^{-1}[y']])$	are measurable $R$-closed.
	Measurability follows from the separability of $A$, $A'$.
	Since $R$ is $z$-closed, $R[R^{-1}[R[x]]]=R[x]$ and $R^{-1}[R[R^{-1}[y']]]=R^{-1}[y']$. 
	Then, $\mu(R^{-1}[R[x]] \cap \supp \mu) = \mu(R^{-1}[R[x]]) = \mu'(R[x])= \mu'(R[x] \cap \supp \mu')$, and similarly for the other equation.
	
	$(\Leftarrow)$ Let $(Q,Q')$ be a measurable $R$-closed pair with $Q\sbq A$ and $Q'\sbq A'$. 
	We can write $Q$ as a disjoint union $Q=\{x\in Q \mid R[x] = \emptyset\} \cup \bigcup_i R^{-1}[R[x_i]]$ for some (countable) subset $\{x_i\}$ of $Q$ (note that, as $R$ is $z$-closed, $R^{-1}\circ R$ is an equivalence relation on its domain).
	
	For every measure $\mu\in \Tfont{T}_a(s)$, we will use the following two facts:
	\begin{enumerate}
		\item $\mu(\{x\in Q \mid R[x] = \emptyset\})=0$. Readily, by the
                  second inclusion of Lemma~\ref{lem:A_s_saturated}, if $R[x]=\emptyset$, then $x\notin A_1(s)$. Hence, $x\notin \supp \mu$.
		\item If $R^{-1}[R[x_i]]\cap \supp \mu \neq \emptyset$, then $\mu(R^{-1}[R[x_i]])= \mu(R^{-1}[R[y_i]])$ for some $y_i$ in the intersection. Furthermore, by $z$-transitivity, $R[x_i]=R[y_i]$.	
	\end{enumerate}		
	We can now calculate:
	\begin{align*}
		\mu(Q)& = \mu(Q\cap \supp \mu) \\
		&= \mu(\{x\in Q \mid R[x] = \emptyset\}\cap \supp \mu) + \mu\bigl(\bigcup_i R^{-1}[R[x_i]]\cap \supp \mu\bigr) \\
		&= \sum_i \mu(R^{-1}[R[x_i]]\cap \supp \mu) \\
		&= \sum \{\mu(R^{-1}[R[x_i]]\cap \supp \mu)\mid R^{-1}[R[x_i]]\cap \supp \mu \neq \emptyset \} \\
		&= \sum \{\mu(R^{-1}[R[y_i]]\cap \supp \mu)\mid R^{-1}[R[x_i]]\cap \supp \mu \neq \emptyset \} \\
                &= \sum \{\mu'(R[y_i]\cap \supp \mu') \mid R^{-1}[R[x_i]]\cap \supp \mu \neq \emptyset \} \\ 
		&= \sum \{\mu'(R[x_i]\cap \supp \mu') \mid R^{-1}[R[x_i]]\cap
                    \supp \mu \neq \emptyset \}
        \end{align*}
        \begin{align*}
	  \phantom{\mu(Q)}
		&= \mu'\bigl(\bigcup \{R[x_i]\cap \supp \mu' \mid R^{-1}[R[x_i]]\cap \supp \mu \neq \emptyset \} \bigr) \\
		&\leq \mu'(Q').
	\end{align*} 
	By symmetry, we conclude $\mu(Q)=\mu'(Q')$.
\end{proof}

\begin{definition}\label{def:uniform-discrete-nlmp}
	A discrete NLMP is called \emph{(finitely supported) uniform}
	if there are partial measurable functions
	$r_{k,n,a}:\{s \in S \mid \Tfont{T}_a(s)\neq\emptyset\} \to [0,1]$ and $t_{k,n,a}:\{s \in S \mid \Tfont{T}_a(s)\neq\emptyset\} \to S$ such that for every $a\in L$, and $s\in S$ satisfying $\Tfont{T}_a(s)\neq\emptyset$,
	\begin{align}\label{eq:discrete-uniform-NLMP}
		\Tfont{T}_a(s)=\bigl\{ \textstyle\sum_k r_{k,n,a}(s)
		\delta_{t_{k,n,a}(s)}\mid n\in \omega \bigr\},
	\end{align}
	and for every $n\in \omega$, $r_{k,n,a}(s) = 0$ for all but finitely
	many $k\in \omega$.
\end{definition}

If $\lmp{S}$ is uniform, then we have a countable family of
functions $t_{k,n,a}$ that can be used to obtain every state that is reachable from a fixed state $s \in S$ in a finite number of transitions.
This will allow us to describe the set $A_s$ of Definition~\ref{def:general_A_s} by means of suitable combinations of the partial functions $t_{k,n,a}$.

Observe that the domains of these functions are measurable sets,
because
\[
\{s \in S \mid \Tfont{T}_a(s) \neq \emptyset\}
= \Tfont{T}^{-1}_a[H_S].
\]
Moreover, the measurability of $t_{k,n,a}$ guarantees that $\{s \in S \mid t_{k,n,a}(s) \in A\} = (t_{k,n,a})^{-1}[A] \in \Sigma$ whenever $A \in \Sigma$.
Hence, if $k, k', n, n' \in \omega$ and $a,a' \in L$, the domain of the composition
$t_{k',n',a'} \circ t_{k,n,a}$ is
\[
\{s \in S \mid t_{k,n,a}(s) \in
\{s \in S \mid \Tfont{T}_{a'}(s) \neq \emptyset\}\}
= (t_{k',n',a'})^{-1}\!\bigl[\Tfont{T}^{-1}_a[H_S]\bigr],
\]
which is measurable. In this way we can form the various compositions of
the partial functions $t_{k,n,a}$; their domains depend on the particular
process and need not all coincide, but the crucial point is that each of
them is measurable on its domain.

\begin{definition}\label{def:x_n}
	If $\lmp{S}$ is a uniform NLMP with enumeration $\{t_{k,n,a} \mid a \in L;\; k,n \in \omega\}$, we define
 	$\{x_n \mid n \ge 1\}$ to be an enumeration of all finite compositions of	the functions $t_{k,n,a}$.  
 	In addition, we set $x_0 = \id$.
\end{definition}

Notice that the domain of the compositions may vary. A key property is the following:

\begin{lemma}\label{lem:A_s_included_in_x_n}
	Let $\mathbb{S} = (S,\Sigma,\{\Tfont{T}_a \mid a \in L\})$ be a uniform NLMP. For $s\in S$, define $X_s \defi\{x_n(s)\}_{n\in \omega}$.
 	Then, $\lmp{X}_s$ is a substructure of $\lmp{S}$ which contains $A_s$.
\end{lemma}
\begin{proof}
The countable set $X_s$ is measurable, and for all $n\in \omega$, $a \in L$, and measure $\mu \in \Tfont{T}_a(x_n(s))$, $X_s$ is a thick subset of $(S, \Sigma, \mu)$. From this we get the substructure property.
		
For the last assertion, we show by induction on $m$ that $A_m(s) \subseteq X_s$.
For the base case, $\{s\} = A_0(s) \subseteq X_s$ because $x_0(s) = s$.
Assume $A_m(s) \subseteq X_s$ and let $z \in A_{m+1}(s) \setminus A_m(s)$.  
Then $z \in \tilde{\Tfont{T}}_a(y)$ for some $y \in A_m(s)$ and $a\in L$, i.e., $z\in \supp \mu$ for some $\mu \in \Tfont{T}_a(y)$.
By the inductive hypothesis $y = x_{n_0}(s)$ for some $n_0 \in \omega$, hence $z = t_{k,m,a}(x_{n_0}(s)) = x_{n_1}(s)$ for certain $k,m,n_1 \in \omega$.
Thus $z \in X_s$, and we conclude that $A_{m+1}(s) \subseteq X_s$.
\end{proof}
We finally arrive at the result that bisimilarity on uniform pointmass processes
is analytic.
\begin{theorem}\label{th:state-bisim-finit-supp}
  State bisimilarity in a finitely supported, uniform NLMP is $\mathbf{\Sigma}_1^1$.
\end{theorem}
\begin{proof}
  Let $\lmp{S} = (S,\Sigma,\{ \Tfont{T}_a \mid a \in L \})$ be a finitely supported, uniform
  NLMP, and for each $s\in S$, consider its countable substructure $\lmp{X}_s$.
  By Lemma~\ref{lem:bisim-equiv-between-generated}, 
  \begin{quote}
    $s \sim_\s s'$ iff there exists a $z$-closed external state bisimulation $R$
    between $\lmp{X}_s$ and $\lmp{X}_{s'}$ such that $s\mathrel{R}s'$.
  \end{quote}
  We will first show that “$R$ is a $z$-closed external state bisimulation” is Borel
  (assuming that testing membership for $R$ is), and afterwards that the existential
  quantifier for $R$ can range on a Cantor space.

  It is easy to see that “$R$ is $z$-closed” is Borel. To be
  an external bisimulation, we need to separate cases in which transition sets
  are empty, for technical reasons:
  \begin{align}
    x \mathrel{R} x'
    \implies
    \forall a \in L, \;
    &\bigl(\Tfont{T}_a(x) = \emptyset \iff \Tfont{T}'_a(x') = \emptyset \bigr)
    \land {} \label{eq:T-empty}\\
    &\bigl(\Tfont{T}_a(x) \neq \emptyset \implies
    \forall \mu \in \Tfont{T}_a(x) \;
    \exists \mu' \in \Tfont{T}'_a(x') \; \mu \mathrel{\bar{R}} \mu' \bigr)
    \land {} \label{eq:zig-line}\\
    &\bigl(\Tfont{T}'_a(x') \neq \emptyset \implies
    \forall \mu \in \Tfont{T}'_a(x') \;
    \exists \mu' \in \Tfont{T}_a(x) \; \mu \mathrel{\bar{R}} \mu' \bigr) \label{eq:zag-line}
  \end{align}
  We look past the first universal quantifier.  The first line
  (\ref{eq:T-empty}) is Borel by hit-measurability; since the second and third
  lines are symmetrical, we focus on the “zig” condition
  \begin{equation}\label{eq:zig}
    \forall \mu \in \Tfont{T}_a(x) \;
    \exists \mu' \in \Tfont{T}'_a(x') \; \mu \mathrel{\bar{R}} \mu' 
  \end{equation}  
  under the assumption that $\Tfont{T}_a(x)$ is nonempty. In this case, it can
  be replaced by the following:
  \begin{equation}
    \label{eq:R-bar-sum}
      \forall n,\,    \exists n',\,
      \bigl( \textstyle\sum_j r_{j,n,a}(x) \delta_{t_{j,n,a}(x)} \bigr)
      \mathrel{\bar{R}}
      \bigl( \textstyle\sum_{j'}r_{j',n',a}(x') \delta_{t_{j',n',a}(x')} \bigr)
  \end{equation}
  We will apply
  Lemma~\ref{lem:barR_characterization_countable_supp}, for $\mu = \sum_j
  r_{j,n,a}(x) \delta_{t_{j,n,a}(x)}$ and analogous $\mu'$, to reduce the occurrence
  of $\bar{R}$ to $R$. We note that 
  \begin{align*}
    \supp \mu &= \{ t_{k,n,a}(x) \mid r_{k,n,a} (x) \neq 0 \}\\
    \supp \mu' &= \{ t_{k',n',a}(x) \mid r_{k',n',a} (x') \neq 0 \}
  \end{align*}
  Hence that lemma implies that (\ref{eq:R-bar-sum}) is equivalent
  to the assertion that for all $n$, there exists $n'$ such that the conjunction of the following two conditions:
  \begin{multline}
    \label{eq:R-bar-first}
    \forall k,\, r_{k,n,a} (x) \neq 0 \implies \\
    \bigl( \textstyle\sum_j r_{j,n,a}(x) \delta_{t_{j,n,a}(x)} \bigr)(R^{-1}[R[t_{k,n,a}(x)]] \cap \supp \mu)
    =\\
    \bigl( \textstyle\sum_{j'}r_{j',n',a}(x') \delta_{t_{j',n',a}(x')} \bigr)(R[t_{k,n,a}(x)] \cap \supp \mu')
  \end{multline}
  and
  \begin{multline}
    \label{eq:R-bar-second}
    \forall k',\, r_{k',n',a} (x') \neq 0 \implies \\
    \bigl( \textstyle\sum_j r_{j,n,a}(x) \delta_{t_{j,n,a}(x)} \bigr)(R^{-1}[t_{k',n',a}(x')] \cap \supp \mu)
    =\\
    \bigl( \textstyle\sum_{j'}r_{j',n',a}(x') \delta_{t_{j',n',a}(x')} \bigr)(R[R^{-1}[t_{k',n',a}(x')]]\cap \supp \mu')
  \end{multline}
  hold.
  
  To see that (\ref{eq:R-bar-first}) is a measurable condition, it is then enough to show that the functions
  \begin{align*}
    G(x,x',R,n,k) &\defi    \bigl( \textstyle\sum_j r_{j,n,a}(x) \delta_{t_{j,n,a}(x)} \bigr)(R^{-1}[R[t_{k,n,a}(x)]] \cap \supp \mu),\\
    G'(x,x',R,n',k) &\defi \bigl( \textstyle\sum_{j'}r_{j',n',a}(x') \delta_{t_{j',n',a}(x')} \bigr)(R[t_{k,n,a}(x)] \cap \supp \mu').
  \end{align*}
  are Borel.
  Let us focus on $G(x,x',R,n,k)$. To analyze its value,
  \begin{equation*}
    \sum\{ r_{j,n,a}(x) \mid t_{j,n,a}(x) \in R^{-1}[R[t_{k,n,a}(x)]] \cap \supp \mu \},
  \end{equation*}
  consider the inner predicate:
  \begin{multline*}
    t_{j,n,a}(x) \in R^{-1}[R[t_{k,n,a}(x)]] \cap \supp \mu \iff \\
    r_{j,n,a}(x)\neq 0 \land \exists z\in X_{x'} \, (t_{k,n,a}(x)\mathrel{R} z \land t_{j,n,a}(x) \mathrel{R} z) \iff \\ 
    r_{j,n,a}(x)\neq 0 \land \exists l, \, (t_{k,n,a}(x)\mathrel{R} x_l(x') \land t_{j,n,a}(x) \mathrel{R} x_l(x')).
  \end{multline*}
  We deduce:
  \begin{equation*}
    G(x,x',R,n',k) = \sum\{ r_{j,n,a}(x) \mid j \in J_{x,x',R,n,k} \},
  \end{equation*}
  where
  \[
    J_{x,x',R,n,k} \defi \{j \mid r_{j,n,a}(x)\neq 0 \land \exists l, \, t_{k,n,a}(x)\mathrel{R} x_l(x') \land t_{j,n,a}(x) \mathrel{R} x_l(x')\}.
  \]

  Since $\lmp{S}$ is finitely supported, the sets $J_{x,x',R,n,k}$ are all
  finite. Each of the conditions $J_{x,x',R,n,k}=N$ for a finite $N\subseteq
  \omega$ is Borel. For example, 
  \begin{multline*}
    J_{x,x',R,n,k}= \{1\} \iff r_{1,n,a}(x)\neq 0 \land \exists l, \, t_{k,n,a}(x)\mathrel{R} x_l(x') \land t_{1,n,a}(x) \mathrel{R} x_l(x') \\
      \land \forall j\neq 1, \, r_{j,n,a}(x) = 0 \lor \neg (t_{k,n,a}(x)\mathrel{R} x_l(x') \land t_{1,n,a}(x) \mathrel{R} x_l(x')).
  \end{multline*}
  We can finally give a Borel definition of $G$:
  \begin{equation*}
    G(x,x',R,n',k) =
    \begin{cases}
      0, & \text{if } J_{x,x',R,n,k} = \emptyset \\
      r_{0,n,a}(x), & \text{if } J_{x,x',R,n,k}= \{0 \} \\
      r_{1,n,a}(x), & \text{if } J_{x,x',R,n,k}= \{1 \} \\
      r_{0,n,a}(x) + r_{1,n,a}(x), & \text{if } J_{x,x',R,n,k}= \{0,1 \} \\
      \dots & \dots
    \end{cases}
  \end{equation*}
  The same analysis can be performed for $G'$, and repeated for
  (\ref{eq:R-bar-second}) as well by using respective functions $K$ and $K'$.

  We finally take advantage of the canonical enumeration of $\lmp{X}_s$ and $\lmp{X}_{s'}$
  and hence replace the unspecified $R$ by a point in $2^{\N\times \N}$.
  Note that $x,x'$ in (\ref{eq:T-empty}) range over elements of ${X}_s$,
  ${X}_{s'}$, respectively. This amounts to $x_p(s)$, $x_{p'}(s')$ for
  $p,p'\in\omega$.
  Therefore, we can define $s\sim_\s s'$ by:
  \begin{align*}
    \exists R\in 2^{\N\times \N},& \; (0,0)\in 
    R \land  \forall (p,p')\in R, 
    \forall a \in L,\\
    &\bigl(\Tfont{T}_a(x_p(s)) = \emptyset \iff \Tfont{T}'_a(x_{p'}(s')) = \emptyset \bigr)
    \land {}\\
    &\bigl(\Tfont{T}_a(x_p(s)) \neq \emptyset \implies
    \forall \mu \in \Tfont{T}_a(x_p(s)) \;
    \exists \mu' \in \Tfont{T}'_a(x_{p'}(s')) \; \mu \mathrel{\bar{R}} \mu' \bigr)
    \land {}\\
    &\bigl(\Tfont{T}'_a(x_{p'}(s')) \neq \emptyset \implies
    \forall \mu \in \Tfont{T}'_a(x_{p'}(s')) \;
    \exists \mu' \in \Tfont{T}_a(x_p(s)) \; \mu \mathrel{\bar{R}} \mu' \bigr).
  \end{align*}
  Since we have already shown that each line of this definition is Borel, we can
  conclude.
\end{proof}

\subsection{Measurable labelled transition systems}
\label{sec:meas-pointm-proc}

In this section we will consider uncountable LTSs by framing them in the context
of discrete NLMPs to be able to obtain definability results. In order to do
this, we assume that the base sets of the LTSs considered are standard Borel
spaces $(S, \Sigma)$, and that the transition relations $R_a\subseteq S\times S$
actually correspond to appropriate measurable functions $\tilde{\Tfont{T}}_a: S
\to \Sigma$.

The way to do this is to regard LTSs over standard Borel spaces as “non-probabilistic NLMPs” \cite{Wolovick}:
These are processes $(S,\Sigma,\{\Tfont{T}_a\mid a\in L\})$ in which $\Tfont{T}_a(s)$ is a set of point, or Dirac, measures.
In this case, we can express $\Tfont{T}_a(s)=\{\delta_x\mid 
x\in \tilde{\Tfont{T}}_a(s)\}$ with $\tilde{\Tfont{T}}_a(s)\subseteq S$ for each $s\in S$. This implies
that we will be working with probability measures only. In particular, the
zero measure is not included, and we get the equivalence $\Tfont{T}_a(s)=\emptyset \iff \tilde{\Tfont{T}}_a(s)=\emptyset$ (cf. Remark~\ref{rem:about_As}(\ref{item:T_empty})).

Since the mapping $\delta:(S,\Sigma)\to(\Delta(S),\Delta(\Sigma))$ given by $\delta(s) = \delta_s$ is an embedding when $S$ is separable and metrizable, we can dispense with the space
$(\Delta(S),\Delta(\Sigma))$ and work with the structure 
$(S,\Sigma,\{\tilde{\Tfont{T}}_a\mid a\in L\})$.

\begin{definition}\label{def:mlts}
	A \emph{measurable labelled transition system}, or 
	MLTS\xindex{MLTS}, is a tuple $\lmp{S} 
	=(S,\Sigma,\{\tilde{\Tfont{T}}_a\mid a\in L\})$ where 
	$(S,\Sigma)$ is a measurable space and for each label 
	$a\in L$, 
	$\tilde{\Tfont{T}}_a:(S,\Sigma)\rightarrow(\Sigma,H(\Sigma))$
	is a measurable map.
\end{definition}

We have the following correspondence between non-probabilistic NLMPs and MLTSs. 
\begin{prop}[{\cite[Prop.~4.7]{Wolovick}}]\label{prop:equiv_nonprobNLMP_MLTS}
	Suppose $\Sigma$ is countably generated and separates points in $S$, and 
	for all $a\in L$ and $s\in S$, $\Tfont{T}_a(s)=\{\delta_x\mid x\in 
	\tilde{\Tfont{T}}_a(s)\}$ for sets 
	$\tilde{\Tfont{T}}_a(s)\subseteq S$. Then 
	$(S,\Sigma,\{\Tfont{T}_a\mid a\in L\})$ is an NLMP if and only if 
	$(S,\Sigma,\{\tilde{\Tfont{T}}_a(s)\mid a\in L\})$ is an MLTS.
\end{prop}

We address the question of the complexity of bisimilarity in this kind of processes.
We can think of an MLTS $\lmp{S}$ as an LTS with certain measurability constraints. 
The transition relations $R_a\subseteq S\times S$ are 
defined by
\begin{equation}\label{eq:mlts-R_a-definition}
	x\mathrel{R_a} y \iff y\in \tilde{\Tfont{T}}_a(x).
\end{equation} 
If $x\mathrel{R_a}y$, we denote it in the usual way for LTS using $x\Rlab{a} y$.
We recall that for LTS, we have the standard relational notion of bisimulation (Definition~\ref{def:LTS-bisimulation}), and thus we can also use it in the context of MLTS.
From \cite[Prop.~12]{2012arXiv1211.0967S}, we know that for an image-countable MLTS over a separable space, state and standard definitions of internal bisimilarity coincide.
Examples of image-uncountable MLTS where these definitions differ can be found in
\cite{DWTC09:qest,D'Argenio:2012:BNL:2139682.2139685}.
 
In the case of external state bisimulation, Equation~\eqref{eq:ext-lift} of lifting a relation yields $\delta_x\mathrel{\bar{R}}\delta_{x'}$ if and only if for every measurable $R$-closed pair $(Q,Q')$, it holds that $x\in Q\iff x'\in Q'$. 
From this we deduce that $R$ is an external state bisimulation if
\begin{center}
  $s\mathrel{R}s'$ implies $\forall x\in \tilde{\Tfont{T}}_a(s) \;
  \exists x'\in \tilde{\Tfont{T}}_a(s') \;
  \forall (Q,Q')\in \Sigma^\times(R) \;
  (x\in Q\iff x'\in Q')$,
\end{center}
and the corresponding “zag” condition hold.

\begin{prop}\label{prop:equality-bisim-MLTS-tradit-ext} Let $\lmp{S}$ and $\lmp{S}'$ be two image-countable MLTS over separable spaces and $R\subseteq S\times S'$ a $z$-closed relation. Then, $R$ is a standard bisimulation if and only if it is an external state bisimulation.
\end{prop}
\begin{proof}
	Suppose that $R \subseteq S \times S'$ is a standard bisimulation, $s \mathrel{R} s'$, and $x \in \tilde{\Tfont{T}}_a(s)$, that is, $s \Rlab{a} x$. 
	Then there exists $x' \in S'$ such that $s' \Rlab{a} x'$ and $x \mathrel{R} x'$. 
	Consequently, if $(Q, Q')$ is a measurable $R$-closed pair, $x \in Q \iff x' \in Q'$. 
	The zag condition is similar, and thus $R$ is an external state bisimulation.
	
	The reverse direction was noted in Remark~\ref{rem:external_is_traditional_countable_NLMP}.\qedhere
\end{proof}

Given this result, we will denote any of these bisimilarities by
$\sim$.

\begin{corollary}\label{cor:s-bisim-t-implies-Ss-bisim-As}
  Let $\lmp{S}$ be an image-countable MLTS over a separable space. For all $s\in S$, $(\lmp{S},s) \sim (\lmp{A}_s,s)$.
\end{corollary}
\begin{proof}
	Proposition~\ref{prop:S,s-bisim-A,s} implies that 
    $(\lmp{S},s) \sim_\s^\times (\lmp{A}_s,s)$, and by
    Proposition~\ref{prop:equality-bisim-MLTS-tradit-ext}, we can replace $\sim_\s^\times$ with standard bisimulation between LTS.
\end{proof}

In what follows, and given the equivalence of all definitions of bisimilarity in
this context, it will
suffice to work with the LTS structure of $\lmp{A}_s$ via the equations
\eqref{eq:mlts-R_a-definition}, that is, with the system $(A_s, \{ R_a \mid a
  \in L \})$. Following the approach in \cite{2012arXiv1211.0967S}, we aim to
represent this system as a tree over $\mathbb{N}$ and thereby find a reduction
of the bisimilarity relation to the isomorphism relation.

\subsection{Uniformly measurable LTS}\label{sec:umlts}

We have seen that Proposition~\ref{prop:equiv_nonprobNLMP_MLTS} relates image-countable MLTSs to a special case of discrete NLMPs.
We still want to uniformly enumerate all states accessible from other states in an MLTS. We could do this by asking that the associated discrete NLMP be uniformly measurable, but we need to do it in a special way.

\begin{definition}\label{def:umlts}
	An image-countable MLTS $\lmp{S} = (S,\Sigma,\{\tilde{\Tfont{T}}_a\mid a\in
	L\})$ is \emph{uniformly measurable} (UMLTS) if the associated discrete NLMP $(S,\Sigma,\{\Tfont{T}_a(s) \mid a\in L\})$ admits a uniform structure such that, for all $a\in L$ and $s\in S$, if $\Tfont{T}_a(s)\neq \emptyset$ ($\tilde{\Tfont{T}}_a(s)\neq \emptyset)$ the enumerating functions $r_{k,n,a}$ and $t_{k,n,a}$ satisfy:
	\begin{enumerate}
		\item $\{t_{0,n,a}(s)\mid n \in \omega\} = \tilde{\Tfont{T}}_a(s) \wedge \forall n,\, r_{0,n,a}(s)=1$; 
		\item $\forall k \geq 1,\, \forall n,\, t_{k,n,a}(s)=s \wedge r_{k,n,a}(s)=0$.
	\end{enumerate}
\end{definition}

Note that, whenever non-empty, we get the equalities 
\begin{align*}
\Tfont{T}_a(s) &= \bigl\{ \textstyle\sum_k r_{k,n,a}(s)
\delta_{t_{k,n,a}(s)}\mid n\in \omega \bigr\} =\{r_{0,n,a}(s) \delta_{t_{0,n,a}(s)}\mid n\in \omega\} \\
&= \{\delta_x\mid x\in \tilde{\Tfont{T}}_a(s)\}.
\end{align*} 
This way, we can obtain a simpler uniform enumeration of $\tilde{\Tfont{T}}_a$.

\begin{lemma}\label{lem:property-UMLTS}
	An MLTS $\lmp{S}=(S,\Sigma,\{\tilde{\Tfont{T}}_a\mid a\in 
	L\})$ is uniformly measurable iff there exist measurable functions $t_{n,a}:\{s \in S \mid \tilde{\Tfont{T}}_a(s)\neq\emptyset\}\to S$ such that $\tilde{\Tfont{T}}_a(s)=\{t_{n,a}(s)\}_{n\in\omega}$ for every $s \in S$ with $\tilde{\Tfont{T}}_a(s)\neq\emptyset$.		
\end{lemma}
\begin{proof}
	If $\lmp{S}$ is uniformly measurable, take $t_{n,a} \defi t_{0,n,a}$.
	Conversely, define $t_{0,n,a}\defi t_{n,a}$ and the rest of $t_{k,n,a}$ and $r_{k,n,a}$ as in Definition~\ref{def:umlts}. 
\end{proof}

We will prove that, for these structures, we can find a Borel reduction from
bisimilarity to 
isomorphism between first-order structures, that is, 
we will show that bisimilarity on UMLTS is classifiable by countable structures.

The reduction is sketched in Figure~\ref{fig:UMLTS-second-method}. We start with a
pointed UMLTS $(\lmp{S},s)$, and we pass to its substructure $\lmp{A}_s$ induced
from $s$, and to the encoding $\calA_s$ of the latter as an element of $\omega\LTS$; both
steps are carried out by the function denoted by $h$. Finally, the
$\omega$-expansion $\Omega$ turns $\calA_s$ into the first-order structure
$(\Tr_{\lmp{S}}(s),\{\Suc_{\lmp{S}}^a\}_{a\in L})$.
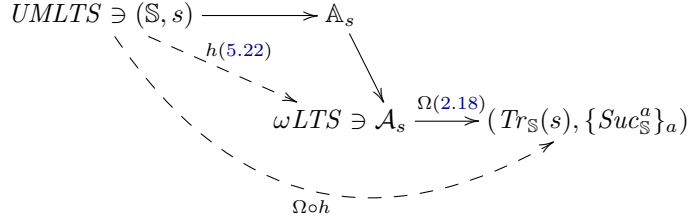
\begin{figure}[h]
  \centerline{
    \xymatrix{
      {\mathit{UMLTS}\ni(\lmp{S},s)} \ar[r]
      \ar@{-->}[dr]^{h(\ref{lem:h:S-to-LTS-continuous})}
      &
           {\lmp{A}_s} \ar[d]!/r 2em/
           &
           \\
           & {\omega \LTS\ni \calA_s}
           \ar[r]^(.45){\Omega (\ref{lem:Omega-continuous})}
           &
              {(\Tr_{\lmp{S}}(s),\{\Suc_{\lmp{S}}^a\}_a)}
    }
  }
  \label{fig:UMLTS-second-method}
  \caption{Obtaining a first-order structure from a pointed UMLTS.}
\end{figure}

\begin{lemma}\label{lem:closed-discrete-implies-UMLTS}
	Let $\mathbb{S} = (S,\Sigma,\{\tilde{\Tfont{T}}_a \mid a \in L\})$ be an
	MLTS and assume that there is a Polish topology $\tau$ on $S$ such that
	$\tilde{\Tfont{T}}_a(s)$ is closed and discrete for every label $a \in L$
	and every $s \in S$. Then $\mathbb{S}$ is uniformly measurable.
\end{lemma}

\begin{proof}
	Because $(S,\tau)$ is Polish, the Selection Theorem for the family of closed sets $F(S)$
	\cite[Thm.~12.13]{Kechris} yields a sequence of Borel functions
	$d_n : F(S) \to S$ such that $\{d_n(F)\}_{n \in \omega}$ is dense in $F$
	for every non-empty closed set $F \subseteq S$.
	
	If $\emptyset \neq \tilde{\Tfont{T}}_a(s) \in F(S)$, the set
	$\{d_n(\tilde{\Tfont{T}}_a(s))\}_{n \in \omega}$ is dense in
	$\tilde{\Tfont{T}}_a(s)$.  Since $\tilde{\Tfont{T}}_a(s)$ is discrete, we
	must have
	$\{d_n(\tilde{\Tfont{T}}_a(s))\} = \tilde{\Tfont{T}}_a(s)$.  
	Hence
	\[	
	t_{n,a}(s) \defi d_n(\tilde{\Tfont{T}}_a(s))
	\]
	is measurable, providing the required measurable enumeration.
\end{proof}

\begin{example}
	Every image-finite MLTS on a standard Borel space is uniformly measurable.
\end{example}

\begin{example}\label{exm:F-unif-meas}
	Consider the MLTS
	\[	
	\mathbb{F} =
	(\Tr_\N \times \N^{<\N},
	\Borel(\Tr_\N \times \N^{<\N}),
	\tilde{\Tfont{T}}),
	\]
	where the set of states reachable under the single label is
	\[	
	\tilde{\Tfont{T}}((T,s)) \defi
	\{(T,s') \mid s,s' \in T \text{ and } s \prec s'\},
	\]
	for the relation $s \prec s' \iff \exists n \in \N\; s' = s^\smallfrown n$.
	This process was introduced in \cite[§3.2]{2012arXiv1211.0967S} as an
	example of an MLTS on a Polish space whose bisimilarity relation is
	analytic but not Borel.
	
	The set $\tilde{\Tfont{T}}((T,s)) = \{T\} \times \{s' \in T \mid s \prec s'\}$
	is discrete, because all its points share the first coordinate~$T$ and
	differ on a discrete second coordinate.  It is also closed, being the
	product of two closed sets.  Consequently, by
	Lemma~\ref{lem:closed-discrete-implies-UMLTS} we conclude that
	$\mathbb{F}$ is uniformly measurable.
\end{example}

Next, to each state $s\in S$ we will associate a bisimilar $\omega\LTS$. For
this, let $f^{(s)}:A_s\to \N$ be defined by $f^{(s)}(r)=\min \{n\in \N 
  \mid x_n(s)=r\}$ (this is well defined by Lemma~\ref{lem:A_s_included_in_x_n})
  and let
\begin{equation}
  \calA_s \defi (\N,\{ f^{(s)}[R_a\rest A_s] \}_{a\in L}).
\end{equation}
It can be seen that $f^{(s)}$ induces a bisimulation between the LTS $\lmp{A}_s$
and $\calA_s$; consequently, the $\omega$-expansion of $\calA_s$ at $0$
(Definition~\ref{def:omega-expansion}) will give a 
bisimilar representation of $\lmp{A}_s$ as a (tree) LTS over 
$\N$, which in turn can be considered as a point $x\in \omega 
\LTS$ satisfying $x_*=0$.

Now let $h:S\to \omega \LTS$ be the map that represents $(\calA_s,0)$ canonically as
an element of $\omega \LTS$
\begin{equation}
	h(s) \defi (0,\{ \chi_{f^{(s)}[R_a\rest A_s]} \}_{a\in L}).
\end{equation}

\begin{lemma}\label{lem:h:S-to-LTS-continuous}
  The application $h:S\to \omega \LTS$ is measurable. \qed
\end{lemma}

\begin{lemma}\label{lem:s-bisim-T-hs}
	Let $\mathbb{S}$ be a UMLTS and $s\in S$. 
	Then, 
	\[
	(\lmp{S},s)\sim (\lmp{A}_s,s)\sim (\calA_s,0) = \ltsfrom_{h(s)}.
	\]
\end{lemma}
\begin{proof}
	By Corollary~\ref{cor:s-bisim-t-implies-Ss-bisim-As} we have the first bisimilarity.
	For the second, it suffices to note that the map $f^{(s)}:\lmp{A}_s\to
        \calA_s$ is an embedding of first-order structures. The last part is immediate from the definitions of $h$ and $\ltsfrom_{h(s)}$.
\end{proof}

\begin{lemma}\label{lem:Omega-h-Borel-reduction}
	\begin{enumerate}
		\item $\Omega\circ h$ is a Borel reduction from $\sim$ on $\lmp{S}$ to  isomorphism on
		$2^{M^{<\N}}\times \bigl(\prod_{a\in L}2^{M^{<\N}\times M^{<\N}}\bigr)$.
		\item\label{item:for_equiv} $\Omega_0\circ h$ is a Borel reduction from $\sim$ on $\lmp{S}$ to  $\equiv$ on $\Tr_M$.
	\end{enumerate}
\end{lemma}
\begin{proof}
  By Lemma~\ref{lem:Omega-continuous}
  and Lemma~\ref{lem:h:S-to-LTS-continuous}, both compositions are Borel. 
  By Lemma~\ref{lem:s-bisim-T-hs}, the transitivity of $\sim$ and Definition~\ref{def:omegaLTS} we have:
  \[
  (\lmp{S}, s) \sim (\lmp{S}, t) \iff \ltsfrom_{h(s)} \sim   \ltsfrom_{h(t)} \iff h(s) \sim h(t).
  \]
  Using Lemma~\ref{lem:sim-iff-Omega-cong}, this is equivalent to $\Omega(h(s)) \cong \Omega(h(t))$, which proves the first Item.
  For the second one, use Lemma~\ref{lem:sim-iff-Omega0-equiv}.
\end{proof}

We obtain immediately from the first item:
\begin{theorem}\label{thm:bisim-MLTS-classifiable-countable-struct}
  If $\lmp{S}$ is a UMLTS, bisimilarity on $\lmp{S}$ is
  classifiable by countable structures.
\end{theorem}

We remind the reader that the isomorphism relation between countable structures is analytic
but in general it is not Borel \cite[16.C, 27.D]{Kechris}.
Since in the uniformly measurable MLTS 
$\lmp{F}$ of Example~\ref{exm:F-unif-meas} bisimilarity is not Borel 
\cite[Thm.~27]{2012arXiv1211.0967S}, this is the best that can be said about the bisimilarity relation for this class of MLTS.

An immediate corollary of this classification is the following:

\begin{corollary}
	If $\lmp{S}$ is a UMLTS, then its bisimilarity classes are measurable sets.
\end{corollary}
\begin{proof}
  This follows from Scott's theorem \cite[Thm.16.6]{Kechris}, which ensures that
  isomorphism classes of countable structures are Borel.
\end{proof}

\subsection{Well-founded part of a uniformly measurable LTS}\label{subsect:complexity-bisim-mlts-bounded-rank}

The last result of the paper is analogous to
Theorem~\ref{thm:bisim-wLTSalpha-Borel}, where we classified bisimilarity for
bounded-rank LTS, but in the context of a single UMLTS.

Recall the formulas defined in \eqref{eq:rank-formulas}, which were used to define the rank of states. 
Let $\lmp{S}=(S,\Sigma,\{\tilde{\Tfont{T}}_a\mid a\in L\})$ be a UMLTS considered as an LTS under the relations in \eqref{eq:mlts-R_a-definition}. 

\begin{prop}\label{prop:semantic-is-Borel-MLTS}
	For every $\alpha <\omega_1$, $\sem{\varphi_\alpha}\in \Sigma$.
\end{prop}
\begin{proof}
	We proceed by induction on $\alpha$. The only non-trivial case is 
	that of $\varphi_{\alpha+1}$, for which we will use the measurable 
	structure of $\lmp{S}$: 
	\begin{align*}
		\sem{\varphi_{\alpha+1}}&=\textstyle{\{s\in S\mid 
			(\lmp{S},s)\vDash
			\bigvee_{a\in L}\pos{a}\varphi_\alpha\}}\\
		&=\{s\in S\mid \exists a\in L, \; \exists t\in S, \; s\Rlab{a} t \wedge t\in \sem{\varphi_\alpha}\}\\
		&=\textstyle{\bigcup_{a\in 
				L}\tilde{\Tfont{T}}_a^{-1}[H_{\sem{\varphi_\alpha}}}].
	\end{align*}
	By the induction hypothesis, $\sem{\varphi_\alpha}\in \Sigma$, and the measurability of 
	$\tilde{\Tfont{T}}_a$ ensures that $\sem{\varphi_{\alpha+1}}\in \Sigma$. 
\end{proof}
\begin{remark}
  It would have been desirable to obtain the previous proposition as a
  consequence of the measurability of the NLMP formulas
  \cite[Sect.~5]{D'Argenio:2012:BNL:2139682.2139685}. Unfortunately, we see no
  straightforward translation from $\mathrm{ML}_{\omega_1}$ formulas to NLMP
  ones, since their variants applying to states include neither negation
  nor disjunction, and lack infinitary conjunctions.
\end{remark}

\begin{definition}
	For each $\alpha<\omega_1$, let $S^{\leq\alpha}$ be the set of states $s\in S$ such that $(\lmp{S},s)$, with the LTS structure given by \eqref{eq:mlts-R_a-definition}, is well-founded of rank at most $\alpha$.
\end{definition}
	
\begin{corollary}\label{cor:S_leq_alpha_is_Borel_subset}
  $S^{\leq\alpha}$ is a Borel subset of $S$.
\end{corollary}
\begin{proof}
	$s\in S^{\leq \alpha} \iff (\lmp{S},s)\nvDash
	\varphi_{\alpha+1}\iff 
	s\notin \sem{\varphi_{\alpha+1}}$.
\end{proof}

\begin{theorem}\label{thm:bisim_S_leq_alpha_Borel}
  Bisimilarity on $S^{\leq \alpha}$ is Borel.
\end{theorem}
\begin{proof}
	By Lemma~\ref{lem:Omega-h-Borel-reduction}(\ref{item:for_equiv}), it is enough to show that an appropriate (co)restric\-tion of $\Omega_0\circ h$ reduces $\sim$ on $S^{\leq \alpha}$ to ${\equiv}\rest(\WF_M^{\leq \alpha})$.
	If $s\in S^{\leq \alpha}$, by Lemma~\ref{lem:s-bisim-T-hs} $(\lmp{S},s)\sim \ltsfrom_{h(s)}$. 
	By Proposition~\ref{prop:bisim-implies-formulas-satisf} we have that these two LTSs have the same rank.
	Hence, using Definition \ref{def:omegaLTS-alpha}, we conclude that	$h(s)\in \omega \LTS^{\leq \alpha}$. 
	Finally, Lemma~\ref{lem:Omega0-codomain} implies $\Omega_0(h(s))\in \WF_M^{\leq 
    \alpha}$.  
\end{proof}

\section{Conclusion}
\label{sec:conclusion}

We have studied the definability of the bisimilarity relation in particular classes of
continuous stochastic processes that include internal nondeterminism, viz. NLMPs.
We have focused on the image-countable case.

The finding in \cite{2012arXiv1211.0967S} that state bisimilarity is not Borel
in general classes of NLMPs (implying the absence of a countable measurable
logic characterizing the former) served as the catalyst for this research. 
This was further motivated by the more precise assertion that, in the NLMP
$\lmp{F}$, bisimilarity behaves like the isomorphism relation between countable
structures; viz., it is a proper analytic relation with Borel equivalence
classes. In Example~\ref{exm:F-unif-meas} we were able to explain that the
reason for this is that $\lmp{F}$ belongs to the larger class of uniformly 
measurable LTSs, for which bisimilarity is classifiable by countable structures.

We hoped that this result could be extended to finitely supported uniform NLMP
(or at least for those with singleton supports), but we were unable to do
so; this is our first open question. On the other extreme of the spectrum, we
still do not know if bisimilarity is analytic for general NLMPs. In future work,
we would like to have sharper lower bounds on the complexity of some of the
classes of processes studied here, with emphasis on surpassing the $E_0$ reduction.

We add two comments on the technical side of this research. Firstly, we found that the original definition of
internal bisimulation \cite{DWTC09:qest} as a symmetric relation, which simplified
several aspects (as in their definition of $\Sigma(R)$), turned out to be
inconvenient when studying the interaction with external versions. We believe
that replacing symmetry with the standard zig and zag conditions (by using
$\Sigma^\times(R)$ instead) could provide a more natural framework.

Secondly, substructures continue to play an important role in our work.
They provided the first layer of simplification: Restricting bisimilarity
to the \textit{generated} substructures allowed us to obtain definability for
finitely supported uniform NLMP. But their interaction with bisimulations is
very nuanced; in particular, “going down” (restricting a bisimulation to a
substructure to obtain another one) requires much stronger hypotheses than
“going up” (showing that being a bisimulation is preserved by passing from a
substructure to the ambient space).

We conclude with a wish list concerning the gap between the hypotheses of some of
our theorems and the supporting counterexamples. In
Lemma~\ref{lem:restriction-smooth-is-NLMP-state-bisim}, we showed that the
restriction of a smooth state bisimulation is a bisimulation, and the next
paragraph commented on a counterexample involving a “smooth” relation (of sorts)
but on a non-definable base space; on the other hand,
Example~\ref{exm:counterexm-NLMP-state-restriction} discusses a somewhat milder
base space but with no trace of smoothness. We would like to have a stronger
version of the lemma, or tighter examples. The same goes for
Lemma~\ref{lem:R_descends_to_bisim_between_substruc}, where we are assuming a
discrete process and a measurably generated bisimulation; we do not have any
example that barely misses the hypotheses. We would also have liked to present
an example of a non-uniform MLTS, but we were unable to do so.

\paragraph*{Acknowledgments}
We are very grateful for the three referees' reports. In particular we are
indebted to Referee Z's observation on
Lemma~\ref{lem:restriction-smooth-is-NLMP-state-bisim}, and Referee Y for their
comments on missing references and context, both in the introduction and on
Section~\ref{sec:prelims-on-markov}. Among other things, we thank Referee X for
pointing out several missing definitions.

We especially thank Nancy Moyano (Centro de Investigación y Estudios de Matemática, CIEM-FaMAF) for support during this project.

\bibliographystyle{rsl}
\bibliography{pointmass_processes}

\newpage

\makeatletter
\def\@seccntformat#1{\csname the#1\endcsname.\ }
\makeatother

\appendix

\section{More results on substructures}
\label{sec:more-results-substructures}

We have seen in Equation \eqref{eq:bisim_in_S_iff_in_substructures} that for a
single discrete NLMP, (global) bisimilarity between two states is determined by
(local) bisimilarity between them over any pair of substructures containing those states.
In this appendix, we provide a generalization of this characterization to states belonging to two possibly distinct discrete processes. 
To this end, we employ a standard construction of a sum process, which allows us to reduce binary relations between distinct objects to a single endorelation on their coproduct.

The \emph{disjoint union} $S\oplus S'$ of two sets is the
tagged union $(S\times \{0\})\cup (S'\times\{1\})$,
equipped with the natural inclusion maps $\inl : S\to S\oplus S'$ and $\inr : S' \to S\oplus S'$.

The \emph{sum of two measurable spaces} $(S,\Sigma)$ and
$(S',\Sigma')$ is the measurable space $(S\oplus S',\Sigma\oplus \Sigma')$, with
the following abuse of $\oplus$: The coproduct $\sigma$-algebra is defined as $\Sigma\oplus\Sigma'\defi\{Q\oplus
Q'\mid Q \in \Sigma,\  Q'\in \Sigma'\}$.
For any measure $\mu\in \Delta(S)$, we define its pushforward $\mu^{\oplus l}\in \Delta(S\oplus S')$ along the inclusion map by $\mu^{\oplus l}(E\oplus E') \defi \mu(E)$. We define $(\mu')^{\oplus r}\in \Delta(S\oplus S')$ analogously for any $\mu'\in \Delta(S')$.

\begin{definition}\label{def:sum-NLMP}
	Let $\lmp{S}$ and $\lmp{S}'$ be two NLMPs. 
	The \emph{sum} NLMP	$\lmp{S}\oplus \lmp{S'}$ has the measurable space $(S\oplus S',\Sigma \oplus \Sigma')$ as its underlying state space, and its transition set functions are given by $\Tfont{T}^\oplus_a(\inl(s)) = \{\mu^{\oplus l} \mid \mu \in \Tfont{T}_a(s)\}$, $\Tfont{T}^\oplus_a(\inr(s')) = \{(\mu')^{\oplus r} \mid \mu' \in \Tfont{T}'_a(s')\}$.
\end{definition}

Note that the sum of two substructures naturally forms a substructure of the sum process. In particular, this applies to the canonical embedding of any substructure from just one of the summands.

Given a relation $R$ on the sum $S\oplus S'$, its \emph{descent} is:
\begin{equation*}\label{eq:R_times}
	R_\times\defi \{(s,s')\mid \inl(s) \mathrel{R} \inr(s')\}\sbq S\times S'.  
\end{equation*}

If now $R\sbq S\times S'$, we can lift it to a relation on $S\oplus S'$ as follows:
\begin{equation*}
	\liftrel{R}\defi \{(\inl(s),\inr(s'))\mid s\mathrel{R} s'\} 
	\sbq (S\oplus S')\times (S\oplus S').
\end{equation*}

\begin{lemma}[{\cite[Proposition~3.23(1)]{survey_bisim}}]\label{lem:ext-equiv-oplus-closed}
	\begin{enumerate}
		\item\label{item:closed-pairs-in-sum-1} Let $R$ be a relation on $A\oplus A'$. 
		If $E\oplus E'$ is $R$-closed, then the pair $(E,E')$ is $R_\times$-closed.
	      \item\label{item:closed-pairs-in-sum-2} If $R\sbq A\times A'$ and the pair $(E,E')$ is $R$-closed, then $E\oplus E'$ is $\bigl(\liftrel{R} \cup \liftrel{R}^{-1}\bigr)$-closed.%
                \footnote{%
                  The original proposition only states $\liftrel{R}$-closure,
                  but it is easy to see that it also holds for the converse
                  relation.}
	\end{enumerate}	
\end{lemma}

\begin{lemma}\label{lem:ext-equiv-oplus-bisim}
	Let $\lmp{S}$ and $\lmp{S}'$ be two NLMP and $\lmp{A}$, $\lmp{A}'$ two substructures of $\lmp{S}$ and $\lmp{S}'$ respectively. 
	\begin{enumerate}
	\item\label{item:ext-bisim-to-sum} If $R\sbq S\times S'$ is an external state bisimulation, then the symmetrization $\liftrel{R} \cup \liftrel{R}^{-1}$ of $\liftrel{R}$ is an internal bisimulation on $\lmp{S}\oplus \lmp{S'}$.
	\item\label{item:int-in-sum-to-ext} If $R$ is a state bisimulation on $\lmp{A}\oplus \lmp{A'}$ such that $R\sbq \liftrel{R_\times} \cup \bigl(\liftrel{R_\times}\bigr)^{-1}$, then $R_\times$ is an external state bisimulation between $\lmp{A}$ and $\lmp{A'}$.
	\end{enumerate}   
\end{lemma}
Note that the condition $R\sbq \liftrel{R_\times} \cup \bigl(\liftrel{R_\times}\bigr)^{-1}$ is another way of stating that $R$ does not relate points within the same summand.
\begin{proof}
  For \ref{item:ext-bisim-to-sum}, take $R\sbq S\times S'$ as stated and let $(\inl(s),\inr(s'))\in \liftrel{R}$.
  Every element of $\Tfont{T}^\oplus_a(\inl(s))$ is of the form $\mu^{\oplus l}$,
  and hence there exists $\mu' \in \Tfont{T}'_a(s')$ such that $\mu \mathrel{\bar{R}} \mu'$.
  Let $E\oplus E'\sbq S\oplus S'$ be measurable and  $\bigl(\liftrel{R} \cup \liftrel{R}^{-1}\bigr)$-closed. 
  Then, it is also $\liftrel{R}$-closed. By Lemma~\ref{lem:ext-equiv-oplus-closed}(\ref{item:closed-pairs-in-sum-1}) the pair $(E,E')$ is $(\liftrel{R})_\times$-closed.
  It is clear that $(\liftrel{R})_\times=R$. 
  Therefore, $\mu^{\oplus l}(E\oplus E') =\mu(E) = \mu'(E') =(\mu')^{\oplus r}(E\oplus E')$.
  The zag condition is analogous.
  Also, the zig and zag conditions for pairs in $\liftrel{R}^{-1}$ are equivalent to the zag and zig conditions, respectively, for the corresponding pair in $\liftrel{R}$. 
  
  For \ref{item:int-in-sum-to-ext}, suppose that $R$ is an internal bisimulation on $\lmp{A}\oplus \lmp{A'}$ and $(s,s')\in R_\times$. 
  If $\mu \in \Tfont{T}_a(s)$, there exists $\mu' \in \Tfont{T}'_a(s')$ such that $\mu^{\oplus l} \mathrel{\bar{R}} (\mu')^{\oplus r}$.
  Let $(E,E')$ be a measurable $R_\times$-closed pair, by
  Lemma~\ref{lem:ext-equiv-oplus-closed}(\ref{item:closed-pairs-in-sum-2})
  $E\oplus E'$ is measurable $\bigl(\liftrel{R_\times} \cup \liftrel{R_\times}^{-1}\bigr)$-closed.
  The hypothesis about $R$ implies that this set is also $R$-closed.
  Thus, $\mu(E) =\mu^{\oplus l}(E\oplus E') =(\mu')^{\oplus r}(E\oplus E') =\mu'(E')$.
\end{proof}

If $\D \subseteq \Power(S)\times\Power(S')$, define a relation $\Rel^\times(\D) \subseteq S\times S'$ as
\[ s\mathrel{\Rel^\times(\D)} s' \iff \forall (Q,Q') \in \D \; (s\in Q \Leftrightarrow s'\in Q').\]

The composition $\Rel^\times\circ \Sigma^\times$ acts as a closure operator
satisfying the following properties \cite[Lem.~4.26, Cor.~4.27]{survey_bisim}:
\begin{enumerate}
	\item $R \subseteq \Rel^\times(\Sigma^\times(R))$.
	\item $\Sigma^\times(R) =
	\Sigma^\times(\Rel^\times(\Sigma^\times(R)))$.
	\item If $R$ is an external bisimulation, then
	$\Rel^\times(\Sigma^\times(R))$ also is.
\end{enumerate}

\begin{prop}\label{prop:bisim_rest_substructures}
	Let $\lmp{S}$ and $\lmp{S}'$ be two discrete NLMP
	over separable spaces and $\lmp{A}$, $\lmp{A}'$ two substructures of $\lmp{S}$ and $\lmp{S}'$ respectively. 
	Let $s\in A$ and $s'\in A'$.
	If $R\subseteq S\times S'$ is an external state 
	bisimulation between $(\lmp{S},s)$ and $(\lmp{S}',s')$ such that $R = \Rel^\times(\Sigma^\times(R))$, 
	then $R{\downarrow}$ is an external state bisimulation between 
	$(\lmp{A},s)$ and $(\lmp{A}',s')$.
\end{prop}
\begin{proof}
	By Lemma~\ref{lem:ext-equiv-oplus-bisim}(\ref{item:ext-bisim-to-sum}), the relation $R_0 = \liftrel{R} \cup \liftrel{R}^{-1}$ is a state bisimulation on the discrete NLMP $\lmp{S} \oplus \lmp{S}'$ which relates $\inl(s)$ and $\inr(s')$.
	Now consider the coarser state bisimulation $R_1=\Rel(\Sigma(R_0))$, which satisfies $R_1=\Rel(\Sigma(R_1))$.
	
	We can use Lemma~\ref{lem:R_descends_to_bisim_between_substruc} applied to $R_1$ and the substructures $\inl[\lmp{A}]$ and $\inr[\lmp{A}']$ to conclude that $R_1{\downarrow}\sbq \inl[A]\times \inr[A']$ is a $z$-closed external state bisimulation between $(\inl[\lmp{A}],\inl(s))$ and $(\inr[\lmp{A}'],\inr(s'))$.
	Given the isomorphisms $\inl{\lmp{A}} \cong \lmp{A}$ and $\inr{\lmp{A}'} \cong \lmp{A}'$, it is now easy to see that $(R_1{\downarrow})_\times$ is an external state bisimulation between $(\lmp{A},s)$ and $(\lmp{A}',s')$.
	We only need to check that this relation is exactly $R{\downarrow}$.
	Note that $E\oplus E' \in \Sigma(R_0)$ if and only if $(E,E')\in \Sigma^\times(R)$. Therefore, for $(x,y)\in A\times A'$,
	\begin{align*}
	  (x,y) \in (R_1{\downarrow})_\times & \iff  (\inl(x),\inr(y)) \in R_1 \\
          & \iff  \forall E\oplus E' \in \Sigma(R_0), \; (x\in E \iff y\in E') \\
		&\iff \forall (E,E')\in \Sigma^\times(R), \; (x\in E \iff y\in E') \\
		&\iff (x,y) \in \Rel^\times(\Sigma^\times(R)) = R.\qedhere 
	\end{align*}
\end{proof}

\begin{corollary}
	Let $\lmp{S}$ and $\lmp{S}'$ be two discrete NLMP
	over separable spaces and $\lmp{A}$, $\lmp{A}'$ two substructures of $\lmp{S}$ and $\lmp{S}'$ respectively. 
	For $s\in A$ and $s'\in A'$, 
	\[
	(\lmp{S},s)\sim_\s^\times (\lmp{S}',s') \iff (\lmp{A},s) \sim_\s^\times (\lmp{A}',s').
	\]
\end{corollary}
\begin{proof}
	From left to right, apply Proposition~\ref{prop:bisim_rest_substructures} to $\sim_\s^\times$. The other direction follows from Lemma~\ref{lem:ext-in-A-iff-ext-in-S}.
\end{proof}

One might ask: How restrictive is the condition that $R$ be a fixed point of $\Rel^\times \circ \Sigma^\times$? 
The following proposition shows that this condition is satisfied automatically whenever $R$ arises from the restriction of a measurably generated equivalence on the sum process.

\begin{lemma}\label{lem:R-downarrow-property}
	If $R=\Rel(\Sigma(R))$, then
	$R{\downarrow}=\Rel^\times(\Sigma^\times(R{\downarrow}))$.
\end{lemma}
\begin{proof}
	We only need to prove the inclusion $\supseteq$. Suppose that $x\mathrel{\Rel^\times(\Sigma^\times(R{\downarrow}))}y$ but $x\mathrel{\cancel{R{\downarrow}}}y$. Then, given the hypothesis $R\supseteq \Rel(\Sigma(R))$, there exists $Q_{x,y} \in \Sigma(R)$ such that $Q\cap\{x,y\}=\{x\}$.
	But then $(Q_{x,y}\cap A, Q_{x,y}\cap A') \in \Sigma^\times(R{\downarrow})$ can distinguish $x$ and $y$. 
\end{proof}

This lemma cannot be strengthened to the equality $\Rel(\Sigma(R)){\downarrow} = \Rel^\times(\Sigma^\times(R{\downarrow}))$.
Only the $\supseteq$ inclusion is valid. 
For the other one, consider $S=\{1, 2, 3\}$, $A=\{1,2\}$ and $A'=\{3\}$. The relation is $R =\{(1,2),(2,3)\}$. Then $(1,3)\in \Rel(\Sigma(R)){\downarrow}\setminus \Rel^\times(\Sigma^\times(R{\downarrow}))$, because $\Sigma(R) =\{S, \emptyset\}$ and $(\{1\},\emptyset)\in \Sigma^\times(R{\downarrow})$.

The next two lemmas give the same conclusion as Proposition~\ref{prop:bisim_rest_substructures}, but
in other scenarios rather than discrete spaces or measurably generated relations.

\begin{lemma}\label{lem:restriction_bisim_closed_pair}
	Let $\lmp{S}$ and $\lmp{S}'$ be two NLMPs, and let $\lmp{A}$, $\lmp{A}'$ be substructures of $\lmp{S}$ and $\lmp{S}'$ respectively. 
	Let $R\subseteq S\times S'$ be an external state bisimulation between $\lmp{S}$ and $\lmp{S}'$.
	If the pair $(A,A')$ is $R$-closed, then the restriction $R{\downarrow}$ is an external state bisimulation between $\lmp{A}$ and $\lmp{A}'$.
\end{lemma}
\begin{proof}
	Let $(Q\cap A,Q'\cap A')$ be a measurable $R{\downarrow}$-closed pair. Then, $R[Q\cap A]\sbq R[A]\sbq A'$, therefore $R[Q\cap A]=R{\downarrow}[Q\cap A]\sbq Q'\cap A'$.
	Similarly, $R^{-1}[Q'\cap A']=R{\downarrow}^{-1}[Q'\cap A']\sbq Q\cap A$. So we conclude that $(Q\cap A,Q'\cap A')$ is also a measurable $R$-closed pair.
	Hence, $\mu_A(Q\cap A)=\mu(Q\cap A)= \mu'(Q'\cap A')=\mu'_{A'}(Q'\cap A')$.
\end{proof}

\begin{lemma}\label{lem:bisim-restriction-MLTS}
  Let $\lmp{S}$ and $\lmp{S}'$ be two image-countable MLTSs. If $R\subseteq S\times S'$ is a 
  bisimulation between $(\lmp{S},s)$ and $(\lmp{S}',s')$, 
  then the restriction $R{\downarrow} = R \cap (A_s \times A_{s'})$ is a bisimulation between $(\lmp{A}_s,s)$ and $(\lmp{A}_{s'},s')$.
\end{lemma}
\begin{proof}
  Since all notions of bisimulation coincide, we use the standard definition. 
  If $(r, r') \in R{\downarrow}$ and $r \Rlab{a} t$, then $t \in A_r \subseteq A_s$. 
  Since $R$ is a bisimulation, there exists $t' \in S'$ such that $r' \Rlab{a} t'$ and $t \mathrel{R} t'$. 
  Thus, $t' \in A_{r'} \subseteq A_{s'}$, and therefore $(t, t') \in R{\downarrow}$. 
  The reciprocal condition is analogous.
\end{proof}

We end this section by extending Corollary~\ref{cor:bisim-in-NLMP-A}(\ref{item:state-on-A-iff-on-S}) to two other known definitions of bisimulation for NLMPs, namely hit and event bisimulations.

\begin{definition}\label{def:nlmp-bisimulations}
	\begin{enumerate}
		\item \label{def:nlmp_hb}
		A relation $R\sbq S\times S$ is a \emph{hit bisimulation} if it is symmetric and for all $a\in L$, $s\mathrel{R}t$ implies $\forall \xi\in \Delta(\Sigma(R))$, $\Tfont{T}_a(s)\cap \xi\neq \emptyset \iff \Tfont{T}_a(t)\cap\xi\neq \emptyset$.
		
		\item \label{def:nlmp_eb} 
		An \emph{event bisimulation} is a sub-$\sigma$-algebra $\Lambda$ of $\Sigma$ such that the map $\Tfont{T}_a: (S,\Lambda)\to(\Delta(\Sigma),H(\Delta(\Lambda)))$ is measurable for each $a\in L$. We also say that a relation $R$ is an event bisimulation if there exists an event bisimulation $\Lambda$ such that $R=\Rel(\Lambda)$.
	\end{enumerate}
	We say that $s,t\in S$ are hit (event) \emph{bisimilar}, denoted as
	$s\mathrel{\sim_\h}t$ ($s\mathrel{\sim_\e}t$), if there is a hit (event) bisimulation $R$ such that $s\mathrel{R}t$.
\end{definition}

\begin{lemma}\label{lem:bisim-in-NLMP-A-appendix} 
	Let $\lmp{A}$ be a substructure of $\lmp{S}$ and let $R\sbq A\times A$ be a symmetric relation.
	\begin{enumerate}
		\item 
		$R$ is a hit bisimulation on $\lmp{A}$ if and only if it is a	hit bisimulation on 	$\lmp{S}$.
		\item 
		If $\Lambda \sbq \Sigma$ is an event bisimulation on $\lmp{S}$, 
		then $\Lambda \rest A$ is an event bisimulation on $\lmp{A}$.
	\end{enumerate}
\end{lemma}
\begin{proof}
	Below, we always assume that $s\rR t$ (necessarily we have $s,t\in A$).
	\begin{enumerate}
		\item $(\Rightarrow)$  
		Suppose $R$ is a hit bisimulation on $\lmp{A}$ and let 
		$D\in \Delta(\Sigma(R))$.
		\begin{align*}
			\Tfont{T}_a(s)\cap D\neq \emptyset 
			&\iff (\Tfont{T}_a\rest A)(s)\cap (\Delta \iota)^{-1}[D] \neq \emptyset \\
			&\iff (\Tfont{T}_a\rest A)(t)\cap (\Delta \iota)^{-1}[D] \neq \emptyset \\
			&\iff \Tfont{T}_a(t)\cap D\neq \emptyset.
		\end{align*}
		
		$(\Leftarrow)$  
		Apply Lemma~\ref{lem:delta-i-preimage} with $\Gamma\defi\Sigma(R)$ and, 
		using Corollary~\ref{cor:bisim-in-NLMP-A}(\ref{item:sigma-R-rest}), obtain the equality
		\[
		(\Delta \iota)^{-1}[\Delta(\Sigma(R))]=\Delta(\Sigma\rest A(R)).
		\]
		If $R$ is a hit bisimulation on $\lmp{S}$ and 
		$\tilde{D}\in \Delta(\Sigma\rest A(R))$, choose 
		$D\in \Delta(\Sigma(R))$ such that 
		$\tilde{D}=(\Delta \iota)^{-1}[D]$. Then
		\begin{align*}
			(\Tfont{T}_a\rest A)(s)\cap \tilde{D}\neq \emptyset 
			&\iff (\Tfont{T}_a\rest A)(s)\cap (\Delta \iota)^{-1}[D]\neq \emptyset \\
			&\iff \Tfont{T}_a(s)\cap D\neq \emptyset \\
			&\iff \Tfont{T}_a(t)\cap D\neq \emptyset \\
			&\iff (\Tfont{T}_a\rest A)(t)\cap \tilde{D}\neq \emptyset.
		\end{align*}
		
		\item  
		First apply Lemma~\ref{lem:delta-i-preimage} with $\Gamma\defi\Lambda$ 
		to obtain 
		$(\Delta \iota)^{-1}[\Delta(\Lambda)]=\Delta(\Lambda\rest A)$.  
		We must verify that
		\[
		\Tfont{T}_a\rest A:(A,\Lambda\rest A)\to
		(\Delta(\Sigma\rest A),H(\Delta(\Lambda\rest A)))
		\]
		is measurable.  
		Let $\tilde{D}\in \Delta(\Lambda\rest A)$ and choose 
		$D\in \Delta(\Lambda)$ such that 
		$\tilde{D}=(\Delta \iota)^{-1}[D]$. Then
		\begin{align*}
			(\Tfont{T}_a\rest A)^{-1}[H_{\tilde{D}}]
			&=\{s\in A\mid (\Tfont{T}_a\rest A)(s)\in H_{\tilde{D}}\}\\
			&=\{s\in A\mid (\Tfont{T}_a\rest A)(s)\cap \tilde{D}\neq\emptyset\}\\
			&=\{s\in A\mid (\Tfont{T}_a\rest A)(s)\cap (\Delta \iota)^{-1}[D]\neq\emptyset\}\\
			&=\{s\in A\mid \Tfont{T}_a(s)\cap D\neq\emptyset\}\\
			&=(\Tfont{T}_a)^{-1}[H_D]\cap A.
		\end{align*}
		Since $D\in \Delta(\Lambda)$, we have $H_D\in H(\Delta(\Lambda))$, 
		hence $(\Tfont{T}_a)^{-1}[H_D]\in \Lambda$ because $\Lambda$ is an 
		event bisimulation.\qedhere
	\end{enumerate}
\end{proof}

\section{Extended proofs}
\label{sec:missing_proofs}

If $T \in \Tr_N$ and $s \in T$, the \emph{section} tree $T_s$ is 
defined by $t \in T_s \iff s^\smallfrown t \in T$.
Recall that $\mathrm{wf}(T)$ denotes the well-founded part of $T$.
We will also use the following notation when necessary: 
if $u \in \mathrm{wf}(T)$ is of the form $(n)^\smallfrown u'$, 
we denote $u'$ by $\mathrm{tail}(u)$.

\begin{lemma}\label{lem:tree-rank-tail}
	If $T \in \Tr_N$ and $u \in \mathrm{wf}(T)$ is of the form 
	$(n)^\smallfrown u'$, then 
	$\rho_T(u) = \rho_{T_{(n)}}(u')$.
\end{lemma}
\begin{proof}
	We note that if $v\in T$ extends $u$, then 
	$\mathrm{tail}(v)$ extends $\mathrm{tail}(u)$. 
	Conversely,	if $v\in T_{(n)}$ extends $\mathrm{tail}(u)$ then 
	$(n)^\smallfrown v$ extends $u$.
	We prove the result by induction on $\alpha=\rho_T(u)$. 
	For the base case, if $u$ is terminal in $T$, then
	$\mathrm{tail}(u)$ must also be terminal in $T_{(n)}$ by 
	what was just observed. 
	Suppose now that the claim holds for every non-empty word 
	in $T$ of rank $\beta<\alpha$ and let $u=(n)^\smallfrown u'\in T$ 
	such that $\rho_T(u)=\alpha$. 
	Then
	\begin{align*}
		\rho_T(u)&=\sup\{\rho_T(v)+1 \mid v\in T, \; u\subsetneq v\}\\
		&=\sup\{\rho_T((n)^\smallfrown v')+1 \mid v'\in T_{(n)}, \; 
		\mathrm{tail}(u)\subsetneq v'\}\\
		&=\sup\{\rho_{T_{(n)}}(v')+1 \mid v'\in T_{(n)}, \; 
		\mathrm{tail}(u)\subsetneq v'\}\\
		&=\rho_{T_{(n)}}(\mathrm{tail}(u)).\qedhere
	\end{align*}
\end{proof}

\begin{prop}[{Proposition~\ref{prop:WF-tree-formulas}}]\label{prop:WF-tree-formulas-appendix}
	$T \notin \WF^{<\alpha}_N \iff (T, \emptyset) \vDash \psi_\alpha$.
\end{prop}
\begin{proof}
	$(\Rightarrow)$ Suppose first that $T\notin \WF_N$ and 
	let $x=(x_i)_{i\in \omega}$ be an infinite branch. 
	We prove by induction on $\alpha<\omega_1$ that $\forall 
	i\geq 0 \; (T,x\rest i)\vDash \psi_\alpha$: the base case is trivial and 
	the limit case follows from the definition of $\psi_\lambda$ and 
	the IH. 
	Suppose that $\forall i\geq 0 \; (T,x\rest i)\vDash \psi_\alpha$. As 
	$x\rest i \to x\rest(i+1)$ and $(T,x\rest(i+1))\vDash \psi_\alpha$, 
	then $(T,x\rest i)\vDash \bigd\psi_\alpha$; therefore 
	$(T,x\rest i)\vDash \psi_{\alpha+1}$ and the claim is proven. 
	In particular, $(T,\emptyset)=(T,x\rest 0) \vDash \psi_\alpha$ for all $\alpha<\omega_1$.
	
	Let us now consider the case of well-founded trees.
	We prove by induction the following claim: if $T\in \WF_N$ and 
	exists $u\in T$ such that $\rho_T(u)\geq\alpha$, then $(T,\emptyset)\vDash \psi_\alpha$. 
	The case $\alpha=0$ is trivial since $(T,\emptyset)\vDash \top$.
	Suppose that the claim is true for $\alpha$ and let $T\in \WF_N$ 
	and $u\in T$ such that $\rho_T(u)\geq \alpha+1$. 
	Then there exists $k\in N$ such that $v \coloneqq u^\smallfrown (k) \in T$ 
	and $\rho_T(v)\geq \alpha$. 
	If $n_0$ is the first coordinate of $u$ (or $\emptyset$ if $u=\emptyset$), 
	by Lemma~\ref{lem:tree-rank-tail}, $\rho_{T_{(n_0)}}(\mathrm{tail}(v))\geq\alpha$. 
	By IH we have that $(T_{(n_0)},\emptyset)\vDash \psi_\alpha$, and 
	therefore $(T,(n_0))\vDash \psi_\alpha$. 
	Then, $(T,\emptyset)\vDash \bigd\psi_\alpha$ and consequently 
	$(T,\emptyset)\vDash \psi_{\alpha+1}$. 
	For the limit case $\lambda$, if $\rho_T(u)\geq \lambda$, then 
	$\forall \alpha <\lambda \; \rho_T(u)\geq \alpha$ and by IH 
	$\forall \alpha<\lambda \; (T,\emptyset)\vDash \psi_\alpha$. 
	Consequently $(T,\emptyset)\vDash \bigwedge_{\alpha<\lambda}\psi_\alpha =\psi_\lambda$.   
	
	$(\Leftarrow)$ Suppose that $(T,\emptyset)\vDash \psi_\alpha$; 
	we must show that $T\notin \WF_N \vee T\in \WF_N^{\geq \alpha}$. 
	We use induction on $\alpha$ to prove that if $T$ is well-founded 
	and $(T,\emptyset)\vDash \psi_\alpha$, then $T\in \WF_N^{\geq 
		\alpha}$. For the case $\alpha=0$ there is nothing to prove as 
	$\WF_N=\WF_N^{\geq 0}$. Suppose that $(T,\emptyset)\vDash 
	\psi_{\alpha+1}=\bigd\psi_\alpha$. Then there exists $k\in N$ 
	such that $(k)\in T$ and $(T_{(k)},\emptyset)\vDash \psi_\alpha$. 
	Given that $T_{(k)}$ is well-founded if $T$ is, by IH we have that 
	$T_{(k)}\in \WF_N^{\geq\alpha}$. It follows that for all 
	$\beta<\alpha$ there exists $v_\beta \in T_{(k)}$ such that 
	$\rho_{T_{(k)}}(v_\beta)\geq \beta$. If $u_\beta \coloneqq (k)^\smallfrown 
	v_\beta$, then $u_\beta \in T$ and by Lemma~\ref{lem:tree-rank-tail} 
	it holds that $\forall \beta <\alpha \; \rho_T(u_\beta)=\rho_{T_{(k)}}(v_\beta)\geq 
	\beta$. Therefore $\rho_T(\emptyset)\geq \alpha$ and consequently 
	$\rho(T)= \rho_T(\emptyset)+1\geq \alpha +1$.
	
	Finally, for the limit case $\lambda$, we note that if 
	$(T,\emptyset)\vDash \psi_\lambda$ then $\forall 
	\alpha<\lambda \; (T,\emptyset)\vDash \psi_\alpha$ and by 
	IH $\forall \alpha<\lambda \; T\in \WF_N^{\geq \alpha}$. 
	Then, $T\in \WF_N^{\geq \lambda}$.
\end{proof}

\begin{lemma}[{Lemma~\ref{lem:R-bar-interpolation}}]\label{lem:R-bar-interpolation-appendix}
  Let $R$ be an equivalence relation on $S$ such that $R=\Rel(\Sigma(R))$,
  $A,A'\subseteq S$, $\nu, \nu'\in \Delta(S)$ such that $A$, $A'$ are thick for
  $\nu$, $\nu'$ respectively. Assume moreover that there exist $B_i\in
  \Sigma\rest A$  and
  $B_i'\in \Sigma\rest A'$ (for $i\in I$) such that
  \begin{enumerate}
  \item\label{item:1-appendix} $\{B_i \mid i \in I \}$ is a partition in measure of $A$;
  \item\label{item:2-appendix} $\{B_i' \mid i \in I \}$ is a partition in measure of $A'$;
  \item\label{item:3-appendix} for all $i\in I$, $\nu_A(B_i) = \nu'_{A'}(B_i')$; and
  \item\label{item:4-appendix}
    for each $i\in I$ there is an $R$-class $C$ %
    such that
    \begin{enumerate}
    \item\label{item:4a-appendix} $C \cap A \subseteq B_i$ and for every $Q\in\Sigma$, if
      $C\cap A \subseteq Q $, then $\nu_A(Q\cap B_i) = \nu_A(B_i)$; and
    \item\label{item:4b-appendix} analogous to Item~\ref{item:4a-appendix} but with $B_i'$ and $A'$.
    \end{enumerate}
  \end{enumerate}
  Then $\nu_A \mathrel{\overline{R{\downarrow}}} \nu'_{A'}$.
\end{lemma}
\begin{proof}
  Let $(U,U')$ be a measurable $R{\downarrow}$-closed pair. We need to show that
  $\nu_A(U) =\nu'_{A'}(U')$.

  First note that $\{U\cap B_i \mid i \in I \}$ is a partition in measure of $U$ and the
  same holds for $B_i'$, $U'$. We will show that
  \begin{enumerate}[label=(\textsc{\roman*})]
  \item\label{item:6-appendix} $\nu_A(U\cap B_i) >0 \iff \nu'_{A'}(U'\cap B_i')>0$; and
  \item\label{item:7-appendix} $\nu_A(U\cap B_i) >0 \implies \nu_{A}(U\cap B_i) =
  \nu_{A}(B_i)$, and the same with primed variables.
  \end{enumerate}
  Assume $\nu_A(U\cap B_i) >0$. By Item~\ref{item:4a-appendix} (with $Q=B_i\sm U$) we
  conclude that $U \cap C \cap A$ is nonempty; we show that $C\cap A' \subseteq
  U'$. Let $u \in U \cap C \cap A$ be arbitrary and let $u' \in C \cap A'$; we
  have that $u\rR u'$ and hence $u'\in U'$ since $(U,U')$ is an
  $R{\downarrow}$-closed pair. We have:
  \begin{align*}
    \nu'_{A'}(U'\cap B_i')
    & =  \nu'_{A'}(B_i') && \text{by Item~\ref{item:4b-appendix},}\\
    & =  \nu_{A}(B_i) && \text{by Item~\ref{item:3-appendix},}\\
    &\geq\nu_A(U\cap B_i)\\
    &> 0  && \text{by assumption.}
  \end{align*}
  
  We can now calculate as follows:
  \begin{align*}
    \nu_A(U) &= \nu_A\bigl(\union\{U\cap B_i \mid i \in I \}\bigr) && \{U\cap
      B_i\}_{i\in I}
      \text{ partitions,}\\
    &= \sum \{\nu_A(U\cap B_i) \mid i \in I \} &&  \text{by
      Lemma~\ref{lem:partition-sum},} \\
    &= \sum \{\underline{\nu_A(U\cap B_i)} \mid i \in I, \nu_A(U\cap B_i) > 0 \} \\
    &= \sum \{\nu_A(B_i) \mid i \in I, \nu_A(U\cap B_i) > 0 \}  &&  \text{by 
      Item~\ref{item:7-appendix},} \\
    &= \sum \{\nu'_{A'}(B_i') \mid i \in I, \underline{\nu_A(U\cap B_i) > 0} \}  &&
    \text{by Item~\ref{item:3-appendix},}\\
    &= \sum \{\nu'_{A'}(B_i') \mid i \in I, \nu'_{A'}(U'\cap B_i') > 0 \}   &&  \text{by 
      Item~\ref{item:6-appendix},} \\
    &= \sum \{\nu'_{A'}(U' \cap B_i') \mid i \in I, \nu'_{A'}(U'\cap B_i') > 0 \}   &&  \text{by 
      Item~\ref{item:7-appendix},} \\
    &= \sum \{\nu'_{A'}(U' \cap B_i') \mid i \in I \} \\
    &= \nu'_{A'}\bigl(\union\{U'\cap B_i' \mid i \in I \}\bigr) \\
    &= \nu'_{A'}(U'). &&\qedhere
  \end{align*}
\end{proof}

\section{Isomorphism on rank-restricted trees}
\label{sec:isom-rank-restr-trees}
In this section we prove that the restriction to
$\WF_M^{\leq \alpha}$ of the $\equiv$ relation defined in (\ref{eq:equiv-def})
is Borel. The result is analogous to the one for isomorphism of trees in
\cite[Section (1.2)]{friedman_stanley_1989}, but we follow the lines of the
proof in \cite[Theorem~13.2.5]{Gao2008InvariantDS}.

\begin{definition}\label{def:cong-alpha}
	For each $1 \leq \alpha < \omega_1$, we define a relation ${\equiv_\alpha} \subseteq \Tr_M \times \Tr_M$ by 
	\[ T \equiv_\alpha T' \iff T, T' \in \WF_M^{=\alpha} \wedge (T, \{\Suc_a(T)\}_{a \in L}) \cong (T', \{\Suc_a(T')\}_{a \in L}), \]
	where $(u, v) \in \Suc_a(T)$ if and only if $\exists x \in M \; \pi_2(x) = a \wedge v = u^\smallfrown (x)$.
\end{definition}

Since ${\equiv}\rest(\WF_M^{\leq \alpha})$ is the (countable) union of the relations
$\equiv_\beta$ for $\beta<\alpha$, it is enough to show that the latter
relations are Borel (Theorem~\ref{thm:iso-alpha-Borel}).

We write $M^k_{\mathrm{inj}} \subseteq M^k$ to denote the 
set of $k$-tuples without repetition, that is, injective functions 
$k \to M$. 
Also, for $k \geq 1$ and $2 \leq \alpha < \omega_1$, 
we denote by $\mathrm{forth}^\alpha_k(T, T')$ and 
$\mathrm{back}^\alpha_k(T, T')$ the following conditions:
\begin{multline*}
  \mathrm{forth}^\alpha_k(T, T')\colon\\
  T \in \WF_M^{=\alpha} 
  \wedge \forall u \in M^k_{\mathrm{inj}}, \; \bigl( \forall 
  i < k, \, (u_i) \in T \implies \\
  \exists u' \in 
  M^k_{\mathrm{inj}}, \; \forall i < k, \, (u'_i) \in T' \wedge \\
  \wedge 
  \forall i < k, \, \pi_2(u_i) = \pi_2(u'_i) \wedge 
  \forall i < k, \; \exists \beta < \alpha, \, 
  T_{(u_i)} \equiv_\beta T'_{(u'_i)} \bigr),
\end{multline*}
\begin{multline*}
  \mathrm{back}^\alpha_k(T, T')\colon\\
  T' \in \WF_M^{=\alpha} 
  \wedge \forall u' \in M^k_{\mathrm{inj}}, \; \bigl( \forall 
  i < k, \, (u'_i) \in T' \implies \\
  \exists u \in 
  M^k_{\mathrm{inj}}, \; \forall i < k, \, (u_i) \in T \wedge \\ 
  \wedge
  \forall i < k, \, \pi_2(u_i) = \pi_2(u'_i) \wedge 
  \forall i < k, \; \exists \beta < \alpha, \, 
  T_{(u_i)} \equiv_\beta T'_{(u'_i)}  \bigr).
\end{multline*}

With the first of these conditions, we express that for any 
finite family of nodes at the first level of $T$ 
(that is, nodes of length one), there are nodes at the first 
level of $T'$ that imitate them one by one, where imitating 
means that the corresponding sections are isomorphic. 
The second condition is the same but starts from $T'$.

\begin{remark}\label{note:forth-back-implies-rank} If $T, T' \in \Tr_M$ 
	satisfy the second part of the definitions of 
	$\mathrm{forth}^\alpha_1(T, T')$ and $\mathrm{back}^\alpha_1(T, T')$, 
	then we can prove that $T, T' \in \WF_M^{\leq \alpha + 1}$. Indeed: for each $(u) \in T$, there exist $(u') \in T'$ and 
	$\beta_u < \alpha$ such that $T_{(u)} \equiv_{\beta_u} T'_{(u')}$. By 
	the definition of $\equiv_{\beta_u}$, $T_{(u)} \in \WF_M^{=\beta_u}$, and 
	by Lemma~\ref{lem:tree-rank-tail},  
	$\rho(T_{(u)}) = \rho_{T_{(u)}}(\emptyset) + 1 = \rho_T((u)) + 1$. Clearly, 
	$T \in \WF_M$ since all trees $T_{(u)}$ with $(u) \in T$ 
	are well-founded. Then,
	\begin{align*}
		\rho(T) &= \rho_T(\emptyset) + 1 = \sup \{\rho_T((u)) + 1 \mid (u) \in 
		T\} + 1 \\
		&= \sup \{\rho(T_{(u)}) \mid (u) \in T\} + 1 \leq \sup \{\beta_u \mid 
		(u) \in T\} + 1 \leq \alpha + 1.
	\end{align*}
	Analogously for $T'$.
\end{remark}

\begin{lemma}\label{lem:iso-alpha-lemma}
	Given $x \in M$, the function $h_x : \Tr_M \to \Tr_M$ defined by 
	$h_x(T) = T_{(x)}$ is continuous.\qed
\end{lemma}

\begin{lemma}\label{lem:back-forth-iso}
	For every $2 \leq \alpha < \omega_1$, $T \equiv_\alpha T' \iff 
	\forall k \geq 1 (\mathrm{forth}^\alpha_k(T, T') \wedge 
	\mathrm{back}^\alpha_k(T, T'))$.
\end{lemma}
\begin{proof}
  $(\Rightarrow)$ Assume $T \equiv_\alpha T'$. 
  By definition, $T,T' \in \WF_M^{=\alpha}$ and we have an 
  isomorphism $f:(T,\{\Suc_a(T)\}_a)\cong 
  (T',\{\Suc_a(T')\}_a)$. 
  Let $k\geq 1$ and $u=(u_i)\in M^k_{\mathrm{inj}}$ such that 
  $\forall i<k \; (u_i)\in T$. 
  Define $u' = (f((u_i)))_{0\leq i<k}\in M^k$, then 
  $(u'_i)\in T'$ and since $(\emptyset,(u_i)) \in 
  \Suc_a(T) \iff (\emptyset,f((u_i)))\in \Suc_a(T')$ 
  we have $\pi_2(u_i) = \pi_2(u'_i)$. 
  Let $\beta_i\defi\rho(T_{(u_i)})$. 
  By Lemma~\ref{lem:tree-rank-tail}, 
  $\beta_i<\rho(T)=\alpha$. 
  Since the restriction $f_i=f\rest T_{(u_i)}$ is an 
  isomorphism between $T_{(u_i)}$ and $T'_{(u'_i)}$, we have 
  that $T'_{(u'_i)}\in \WF_M^{=\beta_i}$ and hence 
  $T_{(u_i)}\equiv_{\beta_i} T'_{(u'_i)}$. 
  This proves $\mathrm{forth}_k^\alpha(T,T')$. 
  The condition $\mathrm{back}_k^\alpha(T,T')$ is analogous.
  
  $(\Leftarrow)$ Let $T,T'\in \Tr_M$ be nonempty trees 
  such that for all $k\geq 1 (\mathrm{forth}^\alpha_k(T,T') 
  \wedge \mathrm{back}^\alpha_k(T,T'))$. 
  By definition of the conditions, $T,T'\in 
  \WF_M^{=\alpha}$ and it only remains to construct an 
  isomorphism 
  $f:(T,\{\Suc_a(T)\}_a)\cong (T',\{\Suc_a(T')\}_a)$. 
  For each $a\in L$, define $D_a(T)\defi \{n\in \N\mid 
  (n,a,0)\in T\}$ and define $D_a(T')$ analogously. 
  It suffices to give a bijection $g:D_a(T)\to D_a(T')$ such 
  that $T_{(n,a,0)}\cong T_{(g(n),a,0)}$ for every 
  $n\in D_a(T)$. 
  If either of the two sets is infinite, the validity of 
  the conditions $\mathrm{forth}^\alpha_k(T,T')$ and 
  $\mathrm{back}^\alpha_k(T,T')$ for every $k\geq 1$ 
  ensures that both are infinite. 
  In this case, to construct the required bijection we 
  consider the isomorphism types of $T_{(n,a,0)}$.
  The conditions for $k=1$ state that exactly the same 
  types occur in $T$ and $T'$. 
  If for some of these types there are $k$ occurrences in 
  $T$, the condition $\mathrm{forth}^\alpha_k(T,T')$ implies 
  that there are at least $k$ occurrences in $T'$, and 
  conversely using the condition 
  $\mathrm{back}^\alpha_k(T,T')$. 
  If some type occurs infinitely many times in one of the 
  trees, it occurs infinitely many times in the other. 
  In both cases we obtain a bijection between $D_a(T)$ and 
  $D_a(T')$ that preserves isomorphism types.
  
  On the other hand, if the sets are finite, using the 
  condition $\mathrm{forth}^\alpha_k$ with 
  $k=|D_a(T)|$ and $\mathrm{back}^\alpha_{k'}$ with 
  $k'=|D_a(T')|$ we obtain that they have the same 
  cardinality and moreover obtain a bijection satisfying 
  the required property.
\end{proof}

\begin{theorem}\label{thm:iso-alpha-Borel}
	For each \( 1 \leq \alpha < \omega_1 \), \( \equiv_\alpha \) is Borel.
\end{theorem}
\begin{proof}
  We use strong induction on $\alpha$: in the case $\alpha=1$ we note that 
  $(\equiv_1) = \{(\{\emptyset\},\{\emptyset\})\}$ is Borel in 
  $\Tr_M\times \Tr_M$. 
  Now suppose that for every $\beta<\alpha$, $\equiv_\beta$ is Borel. 
  By Lemma~\ref{lem:back-forth-iso}, it suffices to see that the predicates 
  $\mathrm{back}_k^\alpha(T,T')$ and $\mathrm{forth}_k^\alpha(T,T')$ define Borel subsets of $\Tr_M\times \Tr_M$. 
  First note that by Proposition~\ref{prop:WF-alpha-Borel-standard}, the conditions 
  $T,T'\in \WF_M^{=\alpha}$ are Borel. 
  Moreover, all quantifiers range over countable sets and the condition $(u)\in T$ defines a clopen set.
  Finally, the predicate $T_{(u)}\equiv_\beta T'_{(u')}$ is equivalent to 
  $(T,T')\in (h_u\times h_{u'})^{-1}[\equiv_\beta]$ for the function $h_u$ of Lemma~\ref{lem:iso-alpha-lemma}, which defines a Borel set by the inductive hypothesis.
\end{proof}

\end{document}